\documentclass[a4paper]{article}

\usepackage{microtype}

\usepackage[margin=1.2in]{geometry}
\usepackage{amsthm, amssymb, amsmath}
\usepackage{enumitem}
\usepackage{upgreek}

\newtheorem{lemma}{Lemma}

\newtheorem{theorem}{Theorem}

\newtheorem{corollary}{Corollary}
\newtheorem{myproperty}{Property}
\usepackage[boxed]{algorithm}
\usepackage{graphicx}
\usepackage{capt-of}
\usepackage{algorithmicx}
\usepackage{algpseudocode}

\usepackage{thmtools}
\usepackage{thm-restate}

\newcommand{\MT}{\textsf{MaxT}}
\newcommand{\MM}{\textsf{MinR}}
\newcommand{\RAP}{\textsf{RAP}}
\newcommand{\VP}{\textsf{VP}}
\newcommand{\etal}{{\em et al.}}
\newcommand{\remove}[1] {}
\newcommand{\eps}{\varepsilon}
\newcommand{\comment}[1] {}

\newcommand{\negA}{\vspace{-0.05in}}
\newcommand{\negB}{\vspace{-0.1in}}

\newcommand{\mysubsection}[1]{\negB\subsection{#1}\negA}

\newcommand{\myparagraph}[1]{\par\smallskip\par\noindent{\bf{}#1:~}}

\title{\bf The Preemptive Resource Allocation Problem}

\begin{document}
\author{
Kanthi Sarpatwar\thanks{IBM T.J. Watson Research Center, Yorktown Heights, NY. \texttt{sarpatwa@us.ibm.com}}
	\and Baruch Schieber\thanks{Computer Science Department, New Jersey  Institute of Technology, Newark, NJ. \texttt{baruch.m.schieber@njit.edu}}
	\and Hadas Shachnai\thanks{Computer Science Department, Technion, Haifa, Israel. \texttt{hadas@cs.technion.ac.il}}
	}

\maketitle

\begin{abstract}
We revisit a classical scheduling model
to incorporate modern trends in data center networks and cloud services.
Addressing some key challenges in the allocation of shared resources 
to  user requests (jobs) in such settings, we consider the following variants of the classic 
{\em resource allocation problem} ({\RAP}).
The input to our problems is a set $J$ of jobs and a set $M$ of homogeneous hosts, each has available amount of some resource.
A job  is associated with a release time, a due date, a weight and a given length,
as well as its resource requirement.
A \emph{feasible} schedule is an allocation of the resource to a subset of the jobs,
satisfying the job release times/due dates as well as 
the resource constraints.
A crucial distinction between classic {\RAP} and our problems is that we allow preemption and migration of jobs, motivated by virtualization techniques. 

We consider two natural objectives: {\em throughput maximization} ({\MT}), which seeks a maximum weight subset of the jobs that can be feasibly scheduled on the hosts in $M$,
and {\em resource minimization} ({\MM}),
that is finding the minimum number of 
(homogeneous) hosts needed  to feasibly schedule all jobs.
Both problems are known to be NP-hard.
 We first present
 an $\Omega(1)$-approximation algorithm for {\MT} instances where
 time-windows form a laminar family of intervals. We then extend the algorithm to
 handle instances with arbitrary time-windows, assuming there is sufficient slack for each job to be completed. 
 For {\MM} we study a more general setting with $d$ resources and derive an 
$O(\log d)$-approximation for any fixed $d \geq 1$, under the assumption that time-windows are not too small. This assumption can be removed 
leading to a slightly worse ratio of $O(\log d\log^* T)$, where $T$ is the maximum due date of any job.
\end{abstract}



\section{Introduction}
\label{sec:intro}
We revisit a classical scheduling
model to incorporate modern trends in data center networks and cloud services.
The proliferation of virtualization and containerization technologies, along with the advent of increasingly powerful multi-core processors, has made it possible to execute multiple virtual machines (or {\em jobs}) simultaneously on the same host, as well as to preempt and migrate jobs with relative ease.
 We address some fundamental problems in the efficient allocation of shared resources such as CPU cores, RAM, or network bandwidth to several competing jobs.
These problems are modeled to exploit multi-job execution, facilitate preemption and migration while respecting resource and timing constraints. Typically, the infrastructure service providers are oversubscribed and therefore, the common goals here include admission control of jobs to maximize throughput, or minimizing the additional resource required to process all jobs.


The broad setting considered in this paper is the following.
Suppose we are given a set of jobs $J$ that need to be scheduled on a set of identical hosts $M$, where each host has a limited amount of one or more resources. Each job $j \in J$ has release time $r_j$, due date $d_j$, and length $p_j$,
along with a required amount of the resource $s_j$ ($\bar{s}_j$ for multiple resources). A job $j$ can be preempted and migrated across hosts but cannot be processed {\em simultaneously} on multiple hosts, i.e., at any given instant of time a job can be processed by at most one host. However, multiple jobs can be processed by any given host, at any given time, as long as their combined resource requirement does not exceed the available resource. As mentioned above, we consider two commonly occurring objectives, namely, {\em throughput maximization} and {\em resource minimization}.


In the {\em maximum throughput} ({\MT}) variant, we are given a set of homogeneous hosts $M$ and a set of jobs $J$, such that each job $j$ has a profit $w_j$ and attributes
 $(p_j, s_j, r_j,d_j)$. The goal is to find a subset $S\subseteq J$ of jobs of maximum profit $\sum_{j\in S}w_j$ that can be feasibly scheduled on $M$. This problem can be viewed as a preemptive variant of the
classic {\em resource allocation problem} ({\RAP})~\cite{PUW00,chen2002allocation,CCKR11,BF+14}.

In the {\em resource minimization} (\textsf{MinR}) variant, we assume that each job $j$ has a resource requirement vector $\bar{s}_j\in [0,1]^d$ as one of the attributes,
where $d \geq 1$ is the number of available resources. W.l.o.g., we assume that each host has a unit amount of each of the $d$ resources.
A schedule that assigns a set of jobs $S_{i,t}$ to a host $i\in M$ at time $t$ is feasible if
$\sum_{j\in S_{i,t}}\bar{s}_{j} \leq \bar{1}^d$. Given a set of jobs $J$ with attributes $(p_j, \bar{s}_j, r_j,d_j)$, we seek a set of (homogeneous) hosts $M$ of minimum cardinality
 such that all of the jobs can be scheduled feasibly on $M$.
This problem is a generalization of the classic {\em vector packing}
({\VP}) problem, in which a set of $d$-dimensional items needs to be feasibly packed
into a minimum number of  $d$-dimensional bins of unit size in each dimension,
i.e., the vector sum of all the items packed into each bin has to be less than or equal to
$\bar{1}^d$. Any instance of {\VP} can be viewed as an instance of {\MM} with
$r_j=0$, $d_j=1$ and $p_j=1$ for job $j\in J$.

Another application of this general scheduling setting relates to the allocation of space and time to advertisements by online advertisement platforms (such as Google or Facebook).
In the {\em ad placement problem}~\cite{DKS03,FN04} we are given a schedule length of $T$ time slots and a collection of ads that need to be scheduled within this time frame.
The ads must be placed in a rectangular display area whose contents can change in different time slots. All ads share the same height, which is the height of the display area, but may have different widths. Several ads may be displayed simultaneously (side by side), as long as their combined width does not exceed the width of the display area.
In addition, each advertisement specifies a display count (in the range $1, \ldots ,T$), which is the number of time slots during which the ad must be displayed. The actual time slots in which the advertisement will be displayed may be chosen arbitrarily by the scheduler, and, in particular, need not be consecutive.
Suppose that each advertisement is associated with some positive profit, and the scheduler may accept or reject any given ad. A common objective is
to schedule a maximum-profit subset of ads within a display area of given width.
Indeed, this problem can be cast as a special case of {\MT} with a single host, where
all jobs have the same release time and due date.

\mysubsection{Prior Work}

The classical problem of preemptively scheduling a set of jobs with attributes  $(p_j, s_j=1, r_j,d_j)$ on a single machine so as to maximize throughput
can be cast as a special case of {\MT} with a single host, where each job requires all of the available resource. 
Lawler~\cite{L90} showed that in this special case {\MT} admits a PTAS, and
the problem is polynomially solvable for uniform job weights. For multiple hosts (i.e., $m=|M|>1$), this special case of {\MT} ($s_j=1$ for all $j \in J$) admits a $\frac{1}{6+\eps}$-approximation, for any fixed $\eps>0$.
This follows from a result of Kalyanasundaram and Pruhs~\cite{KP01}.

As mentioned earlier, another special case of {\MT} was studied in the context of advertisement placement. The {\em ad placement} problem was introduced by Adler~\etal~\cite{AGM02} and later studied in numerous papers (see, e.g., \cite{DKS03,FN04,DKS05,KDM07,KA+17} and the comprehensive survey in~\cite{PDJ17}).
Freund and Naor~\cite{FN04} presented a $(1/3 - \eps)$-approximation for the maximum profit version, namely, for {\MT} with a single host and the same release time and due date for all jobs.

Fox and Korupula~\cite{fox2013weighted} recently studied our preemptive scheduling model, with job attributes $(p_j, s_j, r_j,d_j)$,
under another popular objective, namely, minimizing weighted flow-time. Their work differs from ours in two ways: while they focus on the online versions, we consider our problems in an offline setting. Further, as they note, while the throughput and resource minimization objectives are also commonly considered metrics, their techniques only deal with flow-time. In fact, these objectives are fundamentally different and we need novel algorithms to tackle them.

The non-preemptive variant of {\MT}, known as the {\em resource allocation problem} (\RAP), was introduced by Phillips~\etal~\cite{PUW00}, and later studied by many authors
   (see, e.g., \cite{BGNS01,BBFNS01,CCKR11,chakaravarthy2014improved,NM+15,chen2002allocation} and the references therein).\footnote{{\RAP} is also known as the {\em bandwidth allocation problem}.}
 Chakaravarthy~\etal~\cite{chakaravarthy2014improved} consider a generalization of  {\RAP} and obtain a constant approximation based on a primal-dual algorithm.
 We note that the preemptive versus non-preemptive problems differ quite a bit in their structural properties.

 As mentioned above, {\MM} generalizes the classic vector packing ({\VP}) problem.
The first non-trivial $O(\log d)$-approximation algorithm for {\VP} was presented by Chekuri and Khanna~\cite{chekuri2004multidimensional}, for any fixed $d \geq 1$.
This ratio was improved by Bansal, Caprara and Sviridenko~\cite{BCS09} to a randomized algorithm with asymptotic approximation ratio arbitrarily close to $\ln d + 1$.  Bansal, Eli\'{a}s and Khan~\cite{bansal2016improved} recently improved this ratio further to $0.807 + \ln(d+1)$.
A ``fractional variant'' of {\MM} problem was considered by Jansen and Porkolab~\cite{JP02}, where time was assumed to be continuous. For this problem, in the case of a single host, they obtain a polynomial time approximation scheme (PTAS), by solving a configuration linear program (rounding the LP solution is not necessary because time is continuous in their case).

Resource minimization was considered also in the context of the ad placement problem.
In this variant, all ads must be scheduled, and the objective is to minimize the width of the display area required to make this possible.
Freund and Naor~\cite{FN04} gave a $2$-approximation algorithm for the problem, which was later improved by Dawande~\etal~\cite{DKS05} to $3/2$.
This implies a $3$-approximation for {\MM} instances with $d=1$, where all jobs have the same release time and due date.
We note that this ratio can be slightly improved, using the property that $s_j \leq 1$ for all $j \in J$.
Indeed, we can schedule the jobs to use the resource, such that
the total resource requirements at any two time slots differ at most by one.
Thus, the total amount of resource
required at any time exceeds the optimum, $OPT$, at most by one unit, implying the jobs
can be feasibly scheduled on $2OPT+1$ hosts.

Another line of work relates to the non-preemptive version of {\MM}, where $d=1$ and
the requirement of each job is equal to $1$ (see, e.g.~\cite{chuzhoy2004machine, chuzhoy2009resource});
thus, at most one job can be scheduled on each host at any time.

\mysubsection{Contributions and Techniques}

Before summarizing our results, we define the notion of {\em slackness}. Denote the time window for processing job $j \in J$ by $\chi_j =[r_j, d_j]$, and let $|\chi_j| = d_j-r_j +1$ denote the length of the interval.
Throughout the discussion, we assume that the
time windows are large enough, namely, there is a  constant $\lambda \in (0,1)$, such that $p_j \leq \lambda |\chi_j|$ for any job $j$.
Such an assumption is quite reasonable in scenarios arising in our applications. We call $\lambda$ the {\em slackness} parameter of the instance.

For the {\MT} problem,
we present (in Section \ref{sec:MaxT}) an $\Omega(1)$ approximation algorithm. As mentioned earlier, the non-preemptive version of this problem is the classic resource allocation problem (\RAP). To see the structural differences between the non-preemptive and preemptive versions, we consider their natural linear programming relaxations. In the case of {\RAP}, it is sufficient to have a single indicator variable $x_{jt}$ for each job $j$ and time slot $t$ to represent its start time. This allows the application of a natural randomized rounding algorithm, where job $j$ is scheduled to start at time $t$ with probability $x_{jt}$. On the other hand, in {\MT}, a job can be preempted several times; therefore, each job requires multiple indicator variables. Further, these variables must be rounded in an {\em all-or-nothing} fashion, i.e., either we schedule all parts of a job or none of them. Our approach to handle this situation is to, somewhat counter-intuitively, ``dumb down'' the linear program  by not committing the jobs to a particular schedule; instead, we choose a subset of jobs that satisfy certain knapsack constraints and construct the actual schedule in a subsequent phase.

We first consider a \emph{laminar} variant of the problem, where the time windows for the jobs are chosen from a laminar family of intervals.\footnote{See the formal definition in 
Section~\ref{sec:prelims}.} This setting includes several important special cases, such as $(i)$ all jobs are released at $t=0$ but have different due dates,
or $(ii)$ jobs are released at different times, but all must  be completed by a given due date. Recall that $m=|M|$ is the number of hosts. Our result for the laminar case is a
$\frac 12 - {\lambda} \left(\frac 12 +\frac{1}{m} \right)$-approximation
algorithm, assuming
that the slackness parameter satisfies $\lambda < {1-\frac{2}{m+2}}$.
Using a simple transformation of an arbitrary instance to laminar,
we obtain a
$\frac 18 - {\lambda} \left( \frac 12 +\frac{1}{m} \right)$-approximation
algorithm for
 general instances, assuming that  $\lambda < \frac 14 - \frac{1}{2(m+2)}$.
 Our results imply that as $\lambda$ decreases, the approximation ratio approaches $\frac{1}{2}$ and $\frac{1}{8}$ for the laminar and the general
 case, respectively. 

Subsequently, we tighten the slackness assumption further to obtain an $\Omega(1)$ approximation algorithm for any constant slackness $\lambda \in (0,1)$ for the laminar case and any constant $\lambda \in (0,\frac{1}{4})$ for the general case. In the special case where the weight of the job is equal to its area, we extend an algorithm due to Chen, Hassin and Tzur~\cite{chen2002allocation} to obtain
 an $\Omega(1)$ approximation guarantee for the general case with no assumption on slackness.

Our algorithm for the laminar case relies on a non-trivial combination of a {\em packing} phase and a {\em scheduling} phase.
While the first phase ensures that the output solution has high profit, the second phase guarantees its feasibility.
To facilitate a successful completion of the selected jobs, we formulate a set of conditions that must
be satisfied in the packing phase.
Both phases make use of the structural properties of a {\em laminar} family of intervals.
In the packing phase, we apply our rounding procedure (for the LP solution) to the tree representation of the intervals.\footnote{This procedure bears some similarity to the {\em pipage} rounding technique of~\cite{AS04}.}
We further use this tree in the scheduling phase, to feasibly assign the resource to the selected jobs in a bottom-up fashion.
Our framework for solving {\MT} is general,
and may therefore find use in other settings of {\em non-consecutive} resource allocation.

%

%
%
For the {\MM} problem, we obtain (in Section~\ref{sec:crcsd})
an $O(\log d)$-approximation algorithm for any constant $d \geq 1$,
 under a mild assumption that any job has a window of size  $\Omega(d^2\log d\log T)$, where $T=\max_{j} d_j$.
We show that this assumption can be removed,
leading to a slight degradation in the approximation factor to $O(\log d\log^* T)$, where $\log^*T$ is the smallest integer $\kappa$ such that
$\underbrace{\log \log \ldots \log }_{\text{$\kappa$ times}}T \leq 1$.
Our  approach builds on a formulation of the problem as a configuration LP, inspired by the works of \cite{BCS09, FGMS11}. However, we quickly deviate from these prior approaches, in order to handle the time-windows and the extra  constraints. Our algorithm involves two main phases:
\emph{a maximization phase} and \emph{residual phase}. Roughly speaking, a configuration is a subset of jobs that can be feasibly assigned to a host at a given time slot $t$.
For each $t$, we choose  $O(m\log d)$ configurations with probabilities proportional to their LP-values.  In this phase, jobs may be allocated the resource only for part of their processing length.
In the second phase, we construct a residual instance based on the amount of time each job has been processed. A key challenge is to show that, for any time window $\chi$, the total ``area'' of jobs left to be scheduled is at most $1/d$
of the original total area. We use this property to solve the residual instance.

\section{Preliminaries}
\label{sec:prelims}
We start with some definitions and notation. For our preemptive variants of {\RAP}, we assume w.l.o.g. that each host has a unit amount of each resource.
We further assume that time is slotted.
We allow non-consecutive allocation of a resource to each job, as
well as job migration. Multiple jobs can be assigned to the same machine at a given time but no job can be processed by multiple machines at the same time. 
Formally, we denote the set of jobs assigned to host $i\in M$ at time $t$ by $S_{i,t}$. We say that job $j$ is {\em completed} if there are $p_j$ time slots
$t\in [r_j, d_j]=\chi_j$ in which $j$ is allocated
its required amount of the resource on some host.
A job $j$ is completed if
$|\{t\in \chi_j: \exists i\in M \mbox{ such that } j\in S_{i,t} \}| \geq p_j$.
Let $T = \max_{j \in J} d_j$ be the latest due date of any job.

 In {\MT}, each job $j \in J$ has a resource requirement $s_j \in (0,1]$. 
 An assignment of  a subset of jobs $S\subseteq J$ to the hosts in $M$ is feasible if each job $j\in S$ is completed,
 and for any time slot $t$ and host $i\in M$, $\sum_{j\in S_{i,t}} s_{j} \leq 1$, i.e., the sum of requirements of all jobs assigned to host $i$ is at most the available resource. 
 
For the {\MM} variant, we assume multiple resources. Thus, each job $j$ has a resource requirement vector $\bar{s}_j\in [0,1]^d$, for some constant $d \geq1$. Further, each host has a unit amount of each of the $d$ resources. An assignment of a set of jobs $S_{i,t}$ to a host $i\in M$ at time $t$ is feasible if $\sum_{j \in S_{i,t}}\bar{s}_{j} \leq \bar{1}^d$.

Let $a_j = s_j p_j$ denote the total resource requirement (or, {\em area}) of job $j\in J$ and refer to the quantity $w_j/a_j$ as  the {\em density} of job $j$. Finally, a set of intervals is {\em laminar} if for any two intervals $\chi'$ and $\chi''$, exactly one of the following holds: $\chi'\subseteq \chi''$, $\chi''\subset \chi'$ or $\chi'\cap \chi'' = \phi$.

\section{Throughput Maximization}
\label{sec:MaxT}


We first consider the case where $\mathcal{L} =  \{\chi_j: j\in J\}$ forms a \emph{laminar} family of intervals. In Section~\ref{sec:MT_laminar}, we present an $\Omega(1)$-approximation algorithm  for the laminar case when
$\lambda \in \left( 0,1-\frac{2}{m+2} \right)$. Following this, we describe (in Section~\ref{sec:MT_general}) our constant approximation for the general case for $\lambda\in \left( 0,\frac 14 -\frac{1}{2(m+2)} \right)$.  We then show, in Section~\ref{sec:MaxT-relax-slack}, how to tighten the results to any constant slackness parameter
(i) $\lambda \in (0,1)$ in the laminar case (ii) $\lambda \in (0,\frac 14)$ in the general case. As an interesting corollary, we obtain an
$\Omega\left( \frac{1}{\log{n}} \right)$-approximation algorithm for the general {\MT} problem with no slackness assumption. Further, we show that in the special case of maximum utilization (i.e., the profit of each job equals its ``area''), we obtain an $\Omega(1)$ guarantee with no assumption on the slackness.


\mysubsection{The Laminar Case}
\label{sec:MT_laminar}
Our algorithm proceeds in two phases. While the first phase ensures that the output solution has high profit, the second phase guarantees its feasibility. Specifically, let
$\omega \in (0, 1 - \frac{\lambda}{m}]$ be a parameter (to be determined).

In Phase 1, we find a subset of jobs $S$ satisfying a knapsack constraint for each $\chi$.
Indeed, any feasible solution guarantees that the total area of jobs within any time-window $\chi \in \mathcal{L}$ is at most $m|\chi|$. Our knapsack constraints further restrict the total area of jobs in $\chi$ to some fraction of $m|\chi|$.
\remove{Note that a simple dynamic program would give us a near optimal solution that satisfies these constraints. However, such an approach provides us with no intuition as to how a near optimal solution for the ``curtailed'' instance compares with the optimal solution of the original instance.}
We adopt an LP-rounding based approach to compute a subset $S$ that is optimal subject to the further restricted knapsack constraints.
(We remark that a dynamic programming approach would work as well. However, such an approach would not provide us with any intuition as to how an optimal solution for the further restricted instance compares with the optimal solution of the original instance.)


In Phase 2 we allocate the resource to the jobs in $S$, by considering separately each host $i$ at a given time slot $t \in [T]$ as a unit-sized bin $(i,t)$ and iteratively assigning
each job $j \in S$ to a subset of such available bins, until $j$ has the resource allocated for $p_j$ distinct time slots.
An outline of the two phases is given in Algorithm~\ref{algm:maxToOutline}.

\begin{algorithm}
\caption{Throughput maximization outline}
\label{algm:maxToOutline}
\begin{algorithmic}[1]
\Require{Set of jobs $J$, hosts $M$ and a parameter $\omega \in (0,1 - \frac{\lambda}{m}]$}
\Ensure{Subset of jobs $S\subseteq J$ and a feasible assignment of $S$ to the hosts in $M$}
\Statex \hspace*{-15pt}{\bf Phase 1:} Select a subset  $S\subseteq J$, such that for each $\chi \in \mathcal{L}$:
  \Statex{$\sum_{j\in S: \chi_j \subseteq \chi} a_j \leq (\omega + \frac{\lambda}{m}) m |\chi|$}
\Statex \hspace*{-15pt}{\bf Phase 2:} Find a feasible allocation of the resource to  the jobs in $S$
\end{algorithmic}
\end{algorithm}

\myparagraph{Phase 1} The algorithm starts by finding a subset of jobs $S \subseteq J$ such that for any $\chi \in \mathcal{L}$: $\sum_{j\in S: \chi_j \subseteq \chi} a_j \leq (\omega + \frac{\lambda}{m}) m |\chi|$.
 We solve the following LP relaxation, in which we impose stricter constraint on the total area of the jobs assigned in each time window $\chi$.
\[
  \begin{array}{llll}
\text{\bf LP:}&  \text{Maximize} &\sum_{j\in J} w_jx_j  & \nonumber\\
 & \text{Subject to:} & \sum_{j: \chi_j \subseteq \chi} a_jx_j \leq \omega m |\chi| & \forall \chi \in \mathcal{L} \nonumber \\
 & & 0\leq x_j \leq 1 & \forall j\in J \nonumber
  \end{array}
  \]

\remove{
\[
\begin{array}{llll}
\text{\bf LP:}&
 \text{Maximize }& \sum_{j\in J}(p_j\alpha_j -\sum_{t\in [T]}\beta_{j,t})& \nonumber\\
& \text{Subject to:} & \sum_{j\in S} (\alpha_j - \beta_{j,t}) - \gamma_t\leq 0, &
\forall C=(S,t) \nonumber  \\
&& \sum_{t\in [T]}\gamma_t \leq 1, & \nonumber \\
&&\alpha_j, \beta_{j,t}, \gamma_t \geq 0&\nonumber\\
&\nonumber
\end{array}
\]
}

\noindent {\em Rounding the Fractional Solution:} Suppose $\mathbf{x}^* = (x_j^*: j\in J)$ is an optimal fractional solution for the LP.
Our goal is to construct an integral solution $\hat{\mathbf{x}} = (\hat{x}_j: j\in J)$. We refer to a job $j$ with
${x}_j^*\in (0,1)$ as a \emph{fractional job}, and to the quantity
$a_j {x}^*_j$ as its \emph{fractional area}.
W.l.o.g., we may assume
that for any interval $\chi\in \mathcal{L}$, there is at most one job $j$ with $\chi_j = \chi$ such that $0< x_j^* <1 $, i.e., it is fractional. Indeed, if  two such jobs exist, then the fractional value of the higher density job (breaking ties arbitrarily) can be increased to obtain a solution no worse than the optimal.  Note, however, that there could be fractional jobs $j'$ with $\chi_{j'} \subset \chi$.

We start by setting $\hat{x}_j = x_j^*$ for all $j\in J$.  Consider the tree
representation of $\mathcal{L}$, which contains a node (also denoted by $\chi$)
for each $\chi \in \mathcal{L}$, and an edge between nodes corresponding to
$\chi$ and $\chi'$, where $\chi'\subset \chi$, if there is no interval
$\chi''\in \mathcal{L}$ such that $\chi' \subset \chi'' \subset \chi$.\footnote{Throughout the discussion we use interchangeably the terms {\em node} and {\em interval} when
referring to a time-window $\chi \in \mathcal{L}$.}
Our rounding
procedure works in a \emph{bottom-up} fashion.
As part of this procedure, we label the nodes with one of two possible colors: \emph{gray} and \emph{black}.
Initially, all leaf nodes are colored \emph{black},
and all internal nodes are colored \emph{gray}.
The procedure terminates when all nodes are colored \emph{black}.
A node $\chi$ is colored as \emph{black} if the following property
holds:

\begin{myproperty}
\label{prop:single_fractional}
For any path $\mathcal{P}(\chi,\chi_l)$ from $\chi$ to a leaf $\chi_l$ there is at most one fractional job $j$ such that $\chi_j$ lies on $\mathcal{P}(\chi,\chi_l)$.
\end{myproperty}

\noindent We note that the property trivially holds for the leaf nodes.
 Now, consider a {\em gray} interval $\chi$ with children $\chi_1, \chi_2, \ldots, \chi_\nu$, each colored {\em black}. Note that $\chi$ is well defined because leaf intervals are all colored black. If there is no fractional job that has $\chi$ as its time-window,
Property~\ref{prop:single_fractional} follows by induction, and we color $\chi$ \emph{black}. Assume  now that $j$ is a fractional job that has $\chi$ as its time-window (i.e., $\chi_j =\chi$). If there is no other fractional job that has its time-window (strictly) contained in $\chi$, Property~\ref{prop:single_fractional} is trivially satisfied. Therefore, assume that there are other fractional jobs $j_1, j_2,\ldots, j_l$ that have their time-windows (strictly) contained in $\chi$.
Now, we decrease the fractional area (i.e., the quantity $a_j \hat{x}_j$) of $j$ by $\Delta$ and increase the  fractional area of jobs in the set
$\{j_1, j_2,\ldots, j_l\}$ by $\Delta_k$ for job $j_k$,
such that $\Delta = \sum_{k\in [l]}\Delta_k$.
Formally, we set $\hat{x}_j\rightarrow \hat{x}_j - \frac{\Delta}{a_j}$ and $\hat{x}_{j_k} \rightarrow \hat{x}_{j_k} + \frac{\Delta_k}{a_{j_k}}$.
We choose these increments such that either $\hat{x}_j$ becomes $0$, or for each $k\in [l]$, $\hat{x}_{j_k}$ becomes $1$. Clearly, in both scenarios, Property~\ref{prop:single_fractional} is satisfied, and we color $\chi$ \emph{black}.

When all nodes are colored {\em black}, we round up the remaining fractional jobs. Namely, for all jobs $j$ such that $\hat{x}_j \in (0,1)$, we set $\hat{x}_j = 1$.  It is important to note that by doing so we may violate the knapsack constraints. However, in Theorem~\ref{thm:rounding}, we bound the violation.

\begin{theorem}
\label{thm:rounding}
Suppose $\mathcal{I} = (J,M,\mathcal{L})$ is a laminar instance of {\em \MT} with optimal profit $W$ and $\forall j\in J$:  $p_j\leq \lambda|\chi_j|$.
For any $\omega \in (0,1-\frac{\lambda}{m}]$, the subset $S= \{j\in J:  \hat{x}_j = 1\}$, obtained as above, satisfies
$\sum_{j\in S} w_j \geq \omega W$, and for any $\chi\in \mathcal{L}$, $\sum_{j\in S:\chi_j\subseteq \chi} a_j \leq (\omega + \frac{\lambda}{m}) m|\chi|$.
\end{theorem}
\begin{proof}
We first observe that any optimal solution $\mathbf{x}^*$ for the LP
satisfies:
$\sum_{j\in J} w_j x_j^* \geq \omega W$. Indeed, consider an optimal solution $O$ for the instance $\mathcal{I}$. We can construct a fractional feasible solution $\mathbf{x}'$ for the LP by setting $x'_j = \omega$ if $j\in O$; otherwise, $x'_j = 0$. Clearly, $\mathbf{x}'$ is a feasible solution for the LP with profit ${\omega W}$.

Consider an integral solution $\hat{\mathbf{x}}$, obtained by applying the rounding procedure on $\mathbf{x}^*$. We first show that $\sum_{j\in J} w_j  \hat{x}_j \geq \omega W$. To this end, we prove that $\sum_{j\in J} w_j \hat{x}_j \geq\sum_{j\in J}w_j x_j^* \geq \omega W$. Suppose we decrease the fractional area of a job $j$ by an amount $\Delta$, i.e., we set $\hat{x}_j \leftarrow \hat{x}_j- \frac{\Delta}{a_j}$.
By the virtue of our procedure, we must simultaneously increase the fractional area of some subset of jobs $F_j$, where for each $k\in F_j$ we have $\chi_k \subset \chi_j$. Further, the combined increase in the fractional area of the jobs in $F_j$ is the same $\Delta$. Now, we observe that the density of job $j$ (i.e., $\frac{w_j}{a_j}$) cannot be higher than any of the jobs in $F_j$. Indeed, if $j'\in F_j$ has density strictly lower than $j$, then the optimal solution $\mathbf{x}^*$ can be improved by decreasing the fractional area of ${j'}$ by some $\epsilon$ while increasing that of $j$ by the same amount (it is easy to see that no constraint is violated in this process) -- a contradiction. Therefore, our rounding procedure will never result in a loss, and $\sum_{j\in J} w_j \hat{x}_j \geq\sum_{j\in J}w_j x_j^* \geq \omega W$.

We now show that, for each $\chi \in \mathcal{L}$, $\sum_{j\in J:\chi_j\subseteq \chi} a_j  \hat{x}_j \leq (\omega + \frac{\lambda}{m}) m|\chi|$. First, observe that for any {\em gray} interval $\chi$ the total fractional area is \emph{conserved}. This is true because there is no transfer of fractional area from the subtree rooted at $\chi$ to
a node outside this subtree until $\chi$ is colored black. Now, consider an interval $\chi$ that is colored \emph{black}. We note that for any job $j$ with $x_j^* =0$, our algorithm ensures that $\hat{x}_j =0$, i.e., it creates no new fractional jobs. Consider the vector $\mathbf{\hat{x}}$ when the interval $\chi$ is converted from \emph{gray} to \emph{black}. At this stage, we have
that the total  (fractional) area packed in the subtree rooted at $\chi$ is
$V(\chi)\stackrel{def}{=}\sum_{j\in J:\chi_j\subseteq \chi} a_j  \hat{x}_j \leq \omega m|\chi|$.
Let $\mathcal{F}(\chi)$ denote the set of all fractional jobs $j'$ that have their time-windows contained in $\chi$ (i.e., $\chi_{j'} \subseteq \chi$).
We claim that the maximum increase in  $V(\chi)$ by the end of the rounding procedure is at most $\sum_{j'\in \mathcal{F}(\chi)} a_{j'}$.
This holds since our procedure does not change the variables $\hat{x}_j \in \{0,1\}$. Thus, the maximum increase in the total area occurs due to rounding all fractional jobs into complete ones,
after all nodes are colored black.

To complete the proof, we now show that  the total area of the fractional jobs in the subtree rooted at $\chi$ satisfies
 $\mathcal{A}[\chi] \stackrel{def}{=} \sum_{j'\in \mathcal{F}(\chi)}a_{j'}
\leq \lambda |\chi|$. We prove this by induction on the level of node $\chi$. Clearly, if $\chi$ is a leaf then the claim holds, since there can exist at most one fractional job $j$ in $\chi$,
 and $a_j \leq p_j \leq \lambda|\chi|$. Suppose  that $\{\chi_1, \chi_2, \ldots \chi_l\}$ are the children of $\chi$. If there is a fractional job $j$ with $\chi_j=\chi$ then, by Property~\ref{prop:single_fractional}, there are no other fractional jobs with time-windows contained in $\chi$. Hence, $\mathcal{A}[\chi] = a_j \leq \lambda|\chi|$. Suppose there is no fractional job
  with $\chi_j=\chi$; then, by
the induction hypothesis: $\mathcal{A}[\chi_k] \leq \lambda |\chi_k|$ for all $k\in [l]$. Further, $\sum_{k\in [l]} |\chi_k| \leq |\chi|$ and $\mathcal{A}[\chi] = \sum_{k\in [l]} \mathcal{A}[\chi_k] \leq \sum_{k\in [l]} \lambda |\chi_k| \leq \lambda |\chi|$.
\end{proof}

Let $O$ be
an optimal solution for $\mathcal{I}$ satisfying
$$\forall \chi\in \mathcal{L}~: ~\sum_{j\in O:\chi_j\subseteq \chi} a_j \leq c m|\chi|,$$ for
some $c \geq 1$. Then it is easy to verify that any optimal solution $\mathbf{x}^*$ for the LP
satisfies:
$\sum_{j\in J} w_j x_j^* \geq \frac{\omega}{c} W$. Hence, we have

\begin{corollary}
\label{cor:thm1}
Suppose $\mathcal{I} = (J,M,\mathcal{L})$ is a laminar instance of {\em \MT}, such that $\forall j\in J:$  $p_j\leq \lambda|\chi_j|$.
Let $S^+ \subseteq J$ be a subset of jobs of total profit $W$ satisfying
$\forall \chi\in \mathcal{L}$: $\sum_{j\in S^+:\chi_j\subseteq \chi} a_j \leq c m|\chi|$, for some $c \geq 1$.
Then, for any $\omega \in (0, 1 - \frac{\lambda}{m}]$, there exists a subset $S \subseteq J$ satisfying
$\sum_{j\in S} w_j \geq \frac{\omega}{c}W$, such that $\forall \chi\in \mathcal{L}$, $\sum_{j\in S:\chi_j\subseteq \chi} a_j \leq (\omega +\frac{\lambda}{m}) m|\chi|$.
\end{corollary}

\myparagraph{Phase 2}
We refer to host $i$ at time $t$ as a \emph{bin} $(i,t)$.
In the course of the allocation phase, we label a \emph{bin} with one of three possible colors: \emph{white}, \emph{gray} or \emph{black}. Initially, all bins are colored \emph{white}.
We color a bin $(i,t)$ gray when some job $j$ is assigned to host $i$ at time $t$, and color it black when we decide to assign no more jobs to this bin.
 Our algorithm works in a bottom-up fashion and marks an interval $\chi$ as \emph{done} when it has successfully completed all the jobs $j$ with $\chi_j\subseteq \chi$. Consider an interval $\chi$ such that any $\chi' \subset \chi$ has already been marked \emph{done}.

Let $j \in J$ be a job with time-window $\chi_j = \chi$, that has not been processed yet. To complete job  $j$, we must pick $p_j$ distinct time slots
 in $\chi$ and assign it to a bin in each slot. Suppose that we have already assigned the job to $p_j' < p_j$ slots so far. Denote by {\em avail}$(j)\subseteq \chi$ the subset of time slots where $j$ has not been assigned yet.
 We pick the next slot and bin as shown in Algorithm~\ref{algm:schedule}.
\negA
\remove{
\begin{algorithm}
\caption{Resource allocation to job $j$ in a single time slot.}
\begin{algorithmic}[1]
\State If there exists a gray bin in  {\em avail}$(j)$ {\bf goto} Step~\ref{label:step2:maxT}; otherwise, pick any white bin $(i,t)$ in \emph{avail}$(j)$ and assign $j$ to host $i$ at time $t$. Color the bin $(i,t)$ gray and {\bf exit}. If no such white bin exists, report {\bf fail}.
\State \label{label:step2:maxT}Let $(i,t)$ be a gray bin in {\em avail}$(j)$, and $S_{(i,t)}$ the set of jobs assigned to this bin. If $\sum_{j'\in S_{(i,t)}}s_{j'} + s_j \leq 1$
 then assign $j$ to host $i$ at time $t$ and {\bf exit}; otherwise, {\bf goto} Step~\ref{label:step3:maxT}.
\State \label{label:step3:maxT}Pick a white bin $(i',t')$ in {\em avail}$(j)$ and assign $j$ to host $i'$ at time $t'$. Color $(i,t)$ and $(i',t')$ as black and pair up $(i,t) \leftrightarrow (i',t')$. If no such white bin exists, report {\bf fail}.
\end{algorithmic}
\end{algorithm}
}
\begin{algorithm}
\caption{Resource allocation to job $j$ in a single time slot}
\label{algm:schedule}
\begin{algorithmic}[1]
\If {there exists a gray bin $(i,t)$ in  {\em avail}$(j)$}
  \State{let $S_{(i,t)}$ be the set of jobs assigned to this bin}
  \If {$\sum_{j'\in S_{(i,t)}}s_{j'} + s_j \leq 1$} \label{label:step2:maxT}
    \State assign $j$ to host $i$ at time $t$
  \ElsIf {there exists a white bin $(i',t')$ in  {\em avail}$(j)$} \label{label:step3:maxT}
    \State assign $j$ to host $i'$ at time $t'$.
    \State color $(i,t)$ and $(i',t')$ black
    \State pair up $(i,t) \leftrightarrow (i',t')$
  \Else \State{report {\bf fail}}
  \EndIf
\ElsIf {there exists a white bin $(i,t)$ in \emph{avail}$(j)$}
  \State{assign $j$ to host $i$ at time $t$}
  \State{color the bin $(i,t)$ gray}
\Else \State{report {\bf fail}}
\EndIf
\end{algorithmic}
\end{algorithm}
\begin{theorem}
\label{thm:group-packing} For any $\lambda < 1-\frac{2}{m+2}$, there exists a
$\frac 12 - {\lambda} \left(\frac 12 +\frac{1}{m} \right)$-approximation
algorithm for the laminar {\em \MT\ }problem, assuming that $p_j \leq \lambda|\chi_j|$ for all $j  \in J$.
\end{theorem}
\begin{proof}
Given an instance $\mathcal{I} = (J, M, \mathcal{L})$ and a parameter $\omega\in (0,1-\frac{\lambda}{m}]$, let $W$ denote the optimal profit. We apply Theorem~\ref{thm:rounding}
 to find a subset of jobs $S\subseteq J$ of profit $\omega W$, such that for any $\chi\in \mathcal{L}$: $\sum_{j\in S: \chi_j\subseteq \chi} a_j \leq (\omega + \frac{\lambda}{m}) m|\chi|$.
We now show that there is a feasible resource assignment to the jobs in $S$ for
$\omega = \frac 12 - {\lambda} \left(\frac 12 +\frac{1}{m} \right)$.
Clearly, this would imply the theorem.

We show that for the above value of $\omega$ Algorithm~\ref{algm:schedule} never reports {\bf fail}, i.e., the resource is feasibly allocated to all jobs in $S$.
 Assume towards contradiction that  Algorithm~\ref{algm:schedule} reports {\bf fail}
while assigning job $j$. Suppose that $j$ was assigned to $p_j' < p_j$ bins before this {\bf fail}.
For $t \in \chi = \chi_j$, we say that bin $(i,t)$
is \emph{bad} if either $(i,t)$ is colored gray, or $j$  has been assigned to some bin $(i',t)$ in the same time slot. We first show that the following invariant holds, as long as no job $j^+$ such that $\chi\subset \chi_{j^+}$ has been allocated the resource: the number of bad bins while processing job $j$ is at most $ \lambda m |\chi|$.   Assuming that the claim is true in each of the child
 intervals of $\chi$, $\{\chi_1, \chi_2\ldots, \chi_l\}$, before any job with time
window $\chi$ is allocated the resource, we have  the number of bad bins = number of gray bins is at most $ \sum_{k\in [l]}\lambda m|\chi_k| \leq \lambda m|\chi|$.
Now, consider the iteration in which we assign $j$ to host $i$ at time $t$. If $(i,t)$ is a gray bin, then the number of bad bins cannot increase. On the other hand, suppose $(i,t)$ was white
before we assign $j$. If there are no gray bins in $\chi$, then the number of bad bins is at most $mp_j \leq \lambda m|\chi|$. Suppose there exist some gray bins, and consider those bins of the form $(i',t')$ such that job $j$ has not been assigned to any host at time $t'$. If there are no such bins, then again the number of bad bins is at most $mp_j \leq \lambda m|\chi|$. Otherwise, we must have considered one such gray bin $(i',t')$ and failed to assign $j$ to host $i'$ at time $t'$. By the virtue of the algorithm, we must have colored both $(i,t)$ and $(i',t')$ black. Thus, the number of bad bins does not increase, and our claim holds.
Now, since we pair the black bins $(i,t)\leftrightarrow (i',t')$ only if $\sum_{j\in S_{(i,t)}}s_j$ + $ \sum_{j'\in S_{(i',t')}}s_{j'} > 1$, the total number of black bins $< 2(\omega + \frac{\lambda}{m}) m|\chi|$. Hence, the total number of bins that are black or bad is $< (\lambda + 2(\omega + \frac{\lambda}{m})) m|\chi|$. Now, setting
$\omega =  \frac 12 - {\lambda} \left(\frac 12 +\frac{1}{m} \right)$, there should be at least one bin $(i^*,t^*)$ that is neither black nor bad. But in this case, we could have assigned $j$ to host $i^*$ at time $t^*$, which is a contradiction to the assumption that the algorithm reports a {\bf fail}.
\end{proof}

For convenience, we restate the claim shown in the proof of Theorem~\ref{thm:group-packing}.
\begin{corollary}
\label{cor:thm2}
Let $\mathcal{I} = (J,M,\mathcal{L})$ be a laminar instance where $p_j\leq \lambda|\chi_j|$ $\forall j\in J$, for
$\lambda \in (0,1)$.
Let $S \subseteq J$ be a subset of jobs, such that for any $\chi\in \mathcal{L}$: $\sum_{j\in S: \chi_j\subseteq \chi} a_j \leq (\omega + \frac{\lambda}{m}) m|\chi|$, where
$\omega \leq \frac 12 - {\lambda} \left(\frac 12 +\frac{1}{m} \right)$. Then, there exists a feasible resource assignment to the jobs in $S$.
\end{corollary}

\mysubsection{The General Case}
\label{sec:MT_general}
We use a simple transformation of general instances of \MT\ into laminar instances and prove an $\Omega(1)$-approximation guarantee. Let $\mathcal{W}$ denote the set of all time-windows for jobs in $J$, i.e., $\mathcal{W} = \{\chi_j: j\in J\}$. We now construct a laminar set of intervals $\mathcal{L}$ and a mapping $\mathfrak{L}:\mathcal{W}\rightarrow \mathcal{L}$.
Recall that $T = \max_{j \in J} d_j$.
The construction is done via a binary tree $\mathcal{T}$ whose nodes  correspond to intervals $[l,r]\subseteq [T]$.
The construction is described in Algorithm~\ref{algm:genToLam}.

\begin{algorithm}
\caption{Transformation into a laminar set\label{algm:genToLam}}
\begin{algorithmic}[1]
\Require{Job set $J$ and $\mathcal{W} = \{\chi_j: j\in J\}$}
\Ensure{Laminar set of intervals $\mathcal{L}$ and a mapping $\mathfrak{L}:\mathcal{W}\rightarrow \mathcal{L}$}
\State let $[T]$ be the root node of tree $\mathcal{T}$
\While {$\exists$ a leaf node $[l,r]$ in $\mathcal{T}$ such that $r -l >1$}
  \State add to $\mathcal{T}$ two nodes $[l, \lfloor \frac{l+r}{2}\rfloor]$
  and $[\lfloor \frac{l+r}{2}\rfloor+1, r]$ as the children of $[l,r]$
\EndWhile
\State let $\mathcal{L}$ be the set of intervals corresponding to the nodes of $\mathcal{T}$
\State For each $\chi\in \mathcal{W}$, let $\mathfrak{L}(\chi) = \chi'$,
where $\chi'$ is the largest interval in $\mathcal{L}$ contained in $\chi$, breaking ties by picking the \emph{rightmost} interval.
\end{algorithmic}
\end{algorithm}

\begin{restatable}{lemma}{laminarlpbound}
\label{lem:laminarlpbound} In Algorithm~\ref{algm:genToLam}, the following properties hold:
\begin{enumerate}
\item For any $j\in J$, $|\chi_j| \leq 4|\mathfrak{L}(\chi_{j})|$.
\item For $\chi\in \mathcal{L}$, let
$\tilde{\chi} = \{t\in \chi_j: j\in J,\, \mathfrak{L}(\chi_j) = \chi\}$, i.e., the union of all time-windows in $\mathcal{W}$ that are mapped to $\chi$.
Then, $|\tilde{\chi}|\leq 4|\chi|$.
\end{enumerate}
\end{restatable}
\begin{proof}
To prove the first property,
it suffices to show that $\chi_j$ cannot completely contain $3$ consecutive intervals in $\mathcal{L}$ that are at the same level as $\mathfrak{L}(\chi_{j})$. Indeed, this would imply that $\chi_j$ cannot intersect more than $4$ consecutive intervals, and therefore  $|\chi_j|\leq 4|\mathfrak{L}(\chi_{j})|$. Now, suppose $\chi_j$ contains at least $3$ such consecutive intervals. Then, by the virtue of our algorithm,  $\mathfrak{L}(\chi_{j})$ is the rightmost interval.  Let $\hat \chi$ be the parent of  $\mathfrak{L}(\chi_{j})$. Two cases arise:

\noindent {\bf Case 1:} $\mathfrak{L}(\chi_{j})$ is a left child of $\hat \chi$.
Consider the two other consecutive intervals at the same level as $\mathfrak{L}(\chi_{j})$ that are contained in $\chi_j$. Observe that these two intervals are siblings; therefore, their parent (which is also in $\mathcal{L}$) is also contained in $\chi_j$. This is a contradiction to the assumption that $\mathfrak{L}(\chi_{j})$ is the largest interval in $\mathcal{L}$ contained in $\chi_j$.

\noindent {\bf Case 2:} $\mathfrak{L}(\chi_{j})$ is a right child of $\hat \chi$. We observe that the sibling of $\mathfrak{L}(\chi_{j})$ must also be contained in $\chi_j$,
implying that
$\hat \chi$
is contained in $\chi_j$, a contradiction.

We now prove the second property.
For any $\chi = [s,d]\in \mathcal{L}$,
let $\chi_l = {[s_l, d_l]}\in \mathcal{W}$ (resp.~$\chi_r = [s_r, d_r]$) be the leftmost (resp.~rightmost) interval
in $\mathcal{W}$ such that $\mathfrak{L}(\chi_l) = \chi$ (resp.~$\mathfrak{L}(\chi_r) = \chi$); then, $\tilde{\chi} = [s_l, d_r]$. Consider the intervals $\chi_1 = [s_l,s]$ and $\chi_2 = [d,d_r]$.
As argued above,
$\chi_l$ cannot contain 3 consecutive  intervals in $\mathcal{L}$ at the same level as $\chi$
Thus, $|\chi_1| < 2|\chi|$.
Also, $|\chi_2| < |\chi|$; otherwise, there is an interval
to the right of
$\chi$ of the same size that can be mapped to $\chi_r$.
Thus, $|\chi_1|+|\chi_2| < 3|\chi|$.
Now, the claim follows by observing that
$|\tilde{\chi}| = |\chi_1| + |\chi|+|\chi_2| \leq 4|\chi|$.
\end{proof}
\begin{theorem}
\label{thm:genToLam} For any $\lambda < \frac 14 -\frac{1}{2(m+2)}$, there exists a
$\frac 18 -\lambda\left( \frac 12 +\frac{1}{m}\right)$-approximation algorithm for
{\em \MT\ }, assuming that $p_j \leq \lambda|\chi_j|$ for all $j  \in J$.
\end{theorem}

\begin{proof}
Given an instance $(J,M,\mathcal{W})$ of \MT\ with slackness parameter $\lambda \in (0,1)$, we first use Algorithm~\ref{algm:genToLam} to obtain a laminar set of intervals $\mathcal{L}$ and the corresponding mapping $\mathfrak{L}:\mathcal{W}\rightarrow \mathcal{L}$.
 Consider a new laminar instance $(J_\ell = \{j_\ell:  j\in J\}, M_\ell = M, \mathcal{L})$, constructed by setting $\chi_{j_\ell} = \mathfrak{L}(\chi_j)$. Note that
  if $S_\ell \subseteq J_\ell$ is a feasible solution for this new instance, the corresponding set $S = \{j: {j_\ell} \in S_{\ell}\}$ is a feasible solution for the original instance.
   Let $\lambda_\ell$ denote the slackness parameter for the new instance. We claim that $\lambda_\ell \leq 4{\lambda}$. Assume this is not true, i.e., there exists a job $j_\ell$, such that $p_{j_\ell} > 4\lambda |\chi_{j_\ell}|$; however, by Lemma~\ref{lem:laminarlpbound}, we have $p_{j_\ell} = p_j \leq \lambda |\chi_j| \leq 4\lambda |\chi_{j_\ell}|$.
A contradiction.
 Now, suppose $O\subseteq J$ is an optimal solution of total profit $W$ for the original (non-laminar) instance. Consider the corresponding subset of jobs $O_\ell = \{j_\ell: j\in O\}$. By
 Lemma~\ref{lem:laminarlpbound}, for any $\chi\in \mathcal{L}$, $|\tilde{\chi}| \leq 4|\chi|$. It follows that, for any $\chi\in \mathcal{L}$, $\sum_{j_\ell\in O_\ell: \chi_{j_\ell}\subseteq \chi}
 a_{j_\ell} = \sum_{j\in O:\mathfrak{L}({\chi_{j}})\subseteq \chi} a_{j} \leq 4m|\chi|$.

Now, we use Corollary~\ref{cor:thm1} for the laminar instance, taking $c=4$, $S^+= O_\ell$ and $\lambda_\ell \in (0,1)$.
Then, for any $\omega \in (0, 1 - \frac{\lambda_\ell}{m}]$, there exists $S_\ell \subseteq J_\ell$ of total profit
$\sum_{j_\ell \in S_\ell} w_j \geq \frac{\omega}{c}W$, such that $\forall \chi\in \mathcal{L}$, $\sum_{j_\ell\in S_\ell:\chi_j\subseteq \chi} a_{j_\ell} \leq (\omega +\frac{\lambda_\ell}{m}) m|\chi|$.
By Corollary~\ref{cor:thm2}, there is a feasible assignment of the
resource to the jobs in $S_\ell$ for
$\omega \leq  \frac 12 -\lambda_\ell\left(\frac 12 +\frac{1}{m}\right)$.
Taking
\[
\omega = \displaystyle{\frac 12 - 4 \lambda\left( \frac 12 + \frac{1}{m}\right)}
\leq
\displaystyle{\frac 12 - \lambda_\ell\left( \frac 12 + \frac{1}{m}\right)},
\]
we have the approximation ratio
$\frac{w}{c} = \frac 18 - 4\lambda \left( \frac 18 + \frac{1}{4m}\right)$, for any
$\lambda < \frac{1}{4} - \frac{1}{2(m+2)}$.

We now return to the original instance and take for the solution the set $S= \{ j:~ j_\ell \in S_\ell \}$.
\end{proof}

\mysubsection{Eliminating the Slackness Requirements}
\label{sec:MaxT-relax-slack}
In this section we show that the slackness requirements in
Theorems~\ref{thm:group-packing} and~\ref{thm:genToLam} can be eliminated,
while maintaining
a constant approximation ratio for \MT\ . In particular, for laminar instances,
we show below that Algorithm~\ref{algm:maxToOutline}
can be used to obtain a polynomial time $\Omega(1)$-approximation for any constant slackness parameter
$\lambda \in (0,1)$. For general \MT\ instances, this leads to an $\Omega(1)$-approximation for any constant $\lambda \in (0,\frac{1}{4})$.
We also show a polynomial time $\Omega(\frac{1}{\log n})$-approximation algorithm for general \MT\ using no assumption on slackness.
We use below the next result, for instances with `large' resource requirement. 

\begin{lemma}
\label{lemma:appx_large_heights}
For any $\delta \in (0,1)$ there is an $\Omega(\frac{1}{\log(1/\delta)})$-approximation for any instance $\mathcal{I} = (J,M,\mathcal{W})$ of \MT\ satisfying $s_j \geq \delta$ $\forall~j \in J$.
\end{lemma}
\begin{proof}
Given an instance $\mathcal{I}$, we first
round down the resource requirement (or, {\em height}) of each job $j \in J$ to the nearest value of the form $\delta(1+\eps')^k$, for some fixed $\eps' \in (0,1)$ and integer
$0 \leq k \leq \lceil \log_{1+\eps'} (\frac{1}{\delta}) \rceil$.
We now partition the jobs into $O(\log (\frac{1}{\delta}))$ classes, such that the jobs in each class have the same rounded height. For a class with
job height $\delta(1+\eps')^k$, let
$m_k =m \cdot  \lfloor \frac{1}{\delta(1+\eps')^k} \rfloor$. We define for this class the instance $\mathcal{I}_k = (J_k,M_k,\mathcal{W})$ of \MT\, in which $|M_k|= m_k$ and $s_j=1$ for all
$j \in J_k$.

Recall that Lawler~\cite{L90} gave a PTAS for \MT\ on a single host, where $s_j=1$ for all $j \in J$. Consider an algorithm ${\mathcal A}_k$ for \MT\ on $\mathcal{I}_k$, which
proceeds as follows. We schedule iteratively the jobs in $J$ on hosts $1, \ldots , m_k$. Let ${\mathcal J}_{i-1}$ be the set of jobs scheduled on hosts $1, \ldots i-1$, and
${\mathcal J}_0 = \emptyset$.
In iteration $i \geq 1$, we use the PTAS of~\cite{L90} for the set of jobs $J \setminus {\mathcal J}_{i-1}$. 
We note that the resulting schedule uses no migrations.
By a result of~\cite{KP01}, this iterative algorithm yields a ratio of $\frac{1}{6+\eps}$ to the 
profit of an optimal schedule for {\MT} (which may use migrations).

Let ${\mathcal A}_k (\mathcal{I}_k)$ be the profit of the solution obtained for $\mathcal{I}_k$. Then we choose the solution set for the instance $\mathcal{I}_{\ell^*}$
which maximizes the profit. That is,
${\mathcal A}_{\ell^*} (\mathcal{I}_{\ell^*})= \max_{0 \leq k \leq \lceil \log_{1+\eps'} (\frac{1}{\delta}) \rceil} {\mathcal A}_k (\mathcal{I}_k)$.
We note that since the job heights are rounded {\em down}, transforming back to the original job heights may require to reduce the number of hosts, at most by factor
$\frac{1}{1+\eps'}$. W.l.o.g.,  assume that $m_{\ell^*} >m$ (otherwise, the rounded height of the scheduled jobs is larger than $\frac{1}{2}$, implying they can be scheduled feasibly
with their original heights on $m$ hosts).
 Thus, among the $m_{\ell^*}$ hosts, we select $m'_{\ell^*}= \lfloor \frac{m_\ell^*}{1+\eps'} \rfloor$ hosts on which the total weight of scheduled jobs is maximized.
It follows  that the approximation ratio is
$(\frac{1}{1+\eps'}- \frac{1}{m_{\ell^*}}) \cdot \frac{1}{(6+\eps) \lceil \log_{1+\eps'} (\frac{1}{\delta}) \rceil}$. 
\end{proof}

\subsubsection{Laminar Instances}
\label{appx_sec:laminar_MaxT}
Recall that $m =|M|$ is the number of hosts.
Given a fixed $\lambda \in (0,1)$, let
\begin{equation}
\label{eq:def_alpha}
\alpha = \alpha(m, \lambda)= \frac{ \lambda(1-\lambda)}{1-\lambda + \frac{\lambda}{m}}.
\end{equation}

In Phase 1 of Algorithm~\ref{algm:maxToOutline}, we round the LP solution to obtain a subset of jobs $S \subseteq J$.
We first prove the following.
\begin{lemma}
\label{lemma: slack_w_small_height}
Let $\lambda \in (0,1)$ be a slackness parameter, and
\begin{equation}
\label{eq:def_omega}
\omega = (1-\alpha)(1-\lambda) - \frac{\alpha \lambda}{m},
\end{equation}
 where $\alpha$ is defined in (\ref{eq:def_alpha}). Then, given a laminar instance $\mathcal{I} = (J,M,\mathcal{L})$ satisfying $p_j \leq \lambda |\chi_j|$ and $s_j \leq \alpha$, there is a feasible allocation of
 the resource to the jobs in S.
\end{lemma}

\begin{proof}
We generate a feasible schedule of the jobs in $S$ proceeding bottom-up in each laminar tree. That is, we start handling job $j$ only once all the jobs $\ell$ with time windows
$\chi_\ell \subset \chi_j$ have been scheduled. Jobs having the same time window are scheduled in an arbitrary order.
Let $j$ be the next job, whose time window is $\chi = \chi_j$. We can view the interval $\chi$ as a set of $|\chi|$ time slots,
each consisting of $m$ unit size bins. We say that a time slot $t \in \chi_j$ is `bad' for job $j$ if there is no space for one processing unit of $j$ (i.e., an `item' of size $s_j$) in any of the bins in $t$;
else, time slot $t$ is `good'.
We note that immediately before we start scheduling job $j$ the number of bad time slots for $j$ is at most
$\frac{m |\chi|(1-\lambda)(1-\alpha) - a_j}{m(1-s_j)}$.
Indeed, by Theorem~\ref{thm:rounding}, choosing for $\omega$ the value in (\ref{eq:def_omega}), after rounding the LP solution the total area of jobs $\ell \in S$, such that $\chi_\ell \subseteq \chi_j$, is at most
\begin{equation}
\label{eq:upper_bound_bad}
(\omega + \frac{\alpha \lambda}{m}) m |\chi| = ((1-\alpha)(1-\lambda) - \frac{\alpha \lambda}{m} + \frac{\alpha \lambda}{m}) m |\chi|.
\end{equation}
In addition, for a time slot $t$ to be `bad' for job $j$, each bin in $t$ has to be at least $(1-s_j)$-full. Hence, the number of good time slots for $j$ is at least
\begin{align}
|\chi| - \frac{m |\chi|(1-\lambda)(1-\alpha) - a_j}{m(1-s_j)} & \!= |\chi| (1 - \frac{(1-\lambda)(1-\alpha)}{1- s_j}) + \frac{a_j}{m(1-s_j)} \nonumber \\
& \! \geq  \frac{p_j}{\lambda} (1 - \frac{(1-\lambda)(1-\alpha)}{1- s_j}) + \frac{a_j}{m(1-s_j)} \nonumber \\
& \!\geq p_j  \nonumber
\end{align}
The first inequality follows from the fact that $p_j \leq \lambda |\chi_j|=\lambda |\chi|$, and the second inequality holds since
$s_j \leq \alpha$. Hence, job $j$ can be feasibly scheduled, for any $j \in S$.
\end{proof}

Using Lemmas~\ref{lemma:appx_large_heights} and~\ref{lemma: slack_w_small_height}, we prove our main result.
\begin{theorem}
\label{thm:MaxT_laminar_tradeoff}
For any $m \geq 1$ and constant $\lambda \in (0,1)$,
\MT\  admits a polynomial time $\Omega(1)$-approximation on any laminar instance $\mathcal{I} = (J,M,\mathcal{L})$ with slackness parameter $\lambda$.
\end{theorem}
\begin{proof}
Given a laminar instance $\mathcal{I}$ satisfying the slackness condition,
we handle separately two subsets of jobs.
\myparagraph{Subset 1}
Jobs $j$ satisfying $s_j \leq \alpha = \alpha(m, \lambda)$, where $\alpha$ is defined in (\ref{eq:def_alpha}). We solve \MT\ for these jobs using Algorithm~\ref{algm:maxToOutline}, taking
the value of $\omega$ as in (\ref{eq:def_omega}). By Theorem~\ref{thm:rounding}, the approximation ratio is $\omega = (1 -\alpha)(1 - \lambda) - \frac{\alpha \lambda}{m}= (1-\lambda)^2$, i.e., we have a constant factor.
\myparagraph{Subset 2}
For jobs $j$ satisfying $s_j > \alpha$, use Lemma~\ref{lemma:appx_large_heights} to obtain an $\Omega(\frac{1}{\log (1/\alpha)})$-approximation.

Taking the best among the solutions for the two subsets of jobs, we obtain an $\Omega(1)$-approximation.
\end{proof}

\subsubsection{The General Case}
\label{appx_sec:general_MaxT}
Recall that, given a general \MT\ instance, $(J,M,\mathcal{W})$, with a slackness parameter $\lambda \in (0,1)$, our transformation yields a new laminar instance
$(J_\ell = \{j_\ell:  j\in J\}, M_\ell = M, \mathcal{L})$ with a slackness parameter
$\lambda_\ell \leq 4 \lambda$ (see the proof of Theorem~\ref{thm:genToLam}).
Now, define
\begin{equation}
\label{eq:def_alpha_ell}
\alpha_\ell = \alpha_\ell(m, \lambda_\ell)= \frac{ \lambda_\ell(1-\lambda_\ell)}{1-\lambda_\ell + \frac{\lambda_\ell}{m}},
\end{equation}
and set
\begin{equation}
\label{eq:def_omega_laminar}
\omega = (1-\alpha_\ell)(1-\lambda_\ell) - \frac{\alpha_\ell \lambda_\ell}{m}.
\end{equation}
Then, by Lemma~\ref{lemma: slack_w_small_height}, we have that any job $j_\ell \in J_\ell$ selected for the solution set $S$ can be assigned the
resource (using Algorithm~\ref{algm:maxToOutline}).

\begin{theorem}
\label{thm:MaxT_general_tradeoff}
For any $m \geq 1$ and constant $\lambda \in (0,\frac{1}{4})$,
\MT\  admits a polynomial time $\Omega(1)$-approximation on any
 instance $\mathcal{I} = (J,M,\mathcal{W})$ with slackness parameter $\lambda$.
\end{theorem}

\begin{proof}
Given such an instance $\mathcal{I}$, consider the resulting laminar instance. As before, we handle separately two subsets of jobs.
\myparagraph{Subset 1}
For jobs $j_\ell \in J_\ell$ satisfying $s_{j_\ell} \leq \alpha_\ell$, where $\alpha_\ell$ is defined in (\ref{eq:def_alpha_ell}), apply Algorithm~\ref{algm:maxToOutline} with
$\omega$ value as in (\ref{eq:def_omega_laminar}). Then, the approximation ratio is $\omega= (1 - \lambda_\ell)^2 \geq (1 -4 \lambda)^2$.
\myparagraph{Subset 2}
For jobs $j_\ell$ where $s_{j_\ell} > \alpha_\ell$,
use Lemma~\ref{lemma:appx_large_heights} to obtain
$\Omega(\frac{1}{\log (1/\alpha_\ell)})$-approximation. 

Taking the best among the solutions for the two subsets of jobs, we obtain an $\Omega(1)$-approximation.
\end{proof}

Finally, consider a general instance of \MT\ . By selecting $\delta = \frac{1}{n}$, we can apply Lemma~\ref{lemma:appx_large_heights} to obtain an
$\Omega(\frac{1}{\log n})$-approximate solution, $S_1$ for the jobs $j \in J$ of heights $s_j \geq \frac{1}{n}$.
Let $S_2$ be a solution consisting of all jobs $j$ for which $s_j < \frac{1}{n}$. Note that this solution is feasible since $\sum_{j\in S_2} s_j < 1$.
Selecting the highest profit solution between $S_1$ and $S_2$, we have the following.
\begin{corollary}
\label{cor:logn_approx_for_MaxT}
There is a polynomial time $\Omega(\frac{1}{\log n})$-approximation algorithm for \MT\ .
\end{corollary}

\subsubsection{Maximizing Utilization} 

Consider instances of {\MT} where the profit gained from scheduling job $j$ is $w_j=a_j=s_j p_j $. 
In this section, we obtain an $\Omega(1)$-approximation for {\MT} instances where the weight of a job is equal to its area. In other words, the goal is to maximize resource utilization.
Our result builds on an algorithm of~\cite{chen2002allocation}.

\begin{theorem}
\label{cor:utility}
There is a polynomial time $\Omega(1)$-approximation for any instance of {\MT} where $w_j=a_j$ for all $j \in J$.
\end{theorem}

\begin{proof}
As before, we represent a time slot $t$ on a host $i$ by a bin $(i,t)$. We first assume that for some constant $\alpha \in (0,1)$, $s_j \leq \alpha$.
For the case where $s_j > \alpha$, we can obtain a constant approximation using Lemma~\ref{lemma:appx_large_heights}.

Fix some $\lambda < \frac 14$.
 We first split the jobs by their lengths. A job $j$ is {\em long} if $p_j > \lambda |\chi_j|$; otherwise, job $j$ is {\em short}.
For a given optimal solution of the problem, let $OPT_s$ and $OPT_\ell$ be the contributions of short jobs and long jobs, respectively.
We handle the long and short jobs separately. Note that since short jobs satisfy the requirements of Theorem~\ref{thm:MaxT_general_tradeoff}, we can obtain a constant approximation with respect to $OPT_s$.

We now handle the long jobs. For this part, we adapt an algorithm due to Chen, Hassin and Tzur~\cite{chen2002allocation}. Suppose $L$ is the set of long jobs.
Consider the following algorithm.

\myparagraph{Step 1} Sort the jobs in non-increasing order of their time-window sizes $|\chi_j|$.
\myparagraph{Step 2} For each job $j$ in the sorted order, if there are $p_j$ time-slots that have at least one bin that is less than $1-s_j$ full, schedule $j$; otherwise, discard it.

Let $A$ be the set of jobs chosen by this algorithm.
We now analyze the performance of the algorithm. For each job $j$, we define an {\em augmented job} $j'$ as follows:
\begin{align*}
&p_{j'} = 3|\chi_j|, \text{ and } s_{j'} = \frac{s_j}{1-\alpha} \nonumber\\
&\chi_{j'} = [r_j-|\chi_j|, d_j + |\chi_j|]\nonumber
\end{align*}

Let $A'$ denote the set of augmented jobs for $A$.
We note that there may be no feasible schedule for the jobs in $A'$.
We define the effective weight of $A'$ as
$$W_{eff}(A') = \sum_{t} \min\{m,\sum_{j'\in A': t\in \chi_{j'}}s_{j'}\}$$

It follows that

$$W(A) \geq \frac{(1-\alpha)\lambda}{3} W(A')\geq \frac{(1-\alpha)\lambda}{3} W_{eff}(A')$$

To complete the proof, we simply show that $W_{eff}(A') \geq W(OPT_\ell)$. To this end, it suffices to show that

$$\sum_{j'\in A': t\in \chi_{j'}}s_{j'} \geq \sum_{o\in OPT_\ell: t\in \chi_o}s_o$$

Two cases arise:

\myparagraph{Case I} No long job $j$ with $t\in \chi_j$ is rejected by our algorithm. In this case the claim follows trivially.
\myparagraph{Case II} There exists some long job $j$ that is rejected by our algorithm. We show that $\sum_{j'\in A': t\in \chi_{j'}}s_{j'} \geq m$. The proof would be follow since $\sum_{o\in OPT_\ell: t\in \chi_o}s_o \leq m$.

Since $j$ was rejected $\exists t'\in \chi_j$ such that each bin $(i,t')$ is at least $(1-\alpha)$ full.
  Let $A_j$ be the set of jobs already scheduled in $t'$ before $j$ was rejected, and let $A_j'$ be the respective set of augmented jobs. We claim that for all $j'_p\in A_j'$, $t\in \chi_{j_p'}$. To see this, we first note that, for any $j_p\in A_j$, we have $|\chi_{j_p}| \geq |\chi_j|$ (because the jobs are chosen in increasing order of time-window sizes). Further, $\chi_{j_p}\cap \chi_j$ contains at least $t'$ and hence $\chi_j$ and $\chi_{j_p}$ are intersecting. Therefore, the augmented job $j_p'\in A_j'$ must completely contain $\chi_j$ and so $t\in \chi_{j_p'}$.
 Thus, we have
 $$\sum_{j'\in A': t\in \chi_{j'}}s_{j'}\geq \sum_{j'_p\in A'_j}s_{j'_p} \geq \frac{1}{1-\alpha}\sum_{j_p\in A_j}s_{j_p}\geq \frac{m(1-\alpha)}{1-\alpha} = m.$$
\end{proof}

\section{Resource Minimization}
\label{sec:crcsd}
In this section, we consider the \MM\ problem with $d$ resources, where $d \geq 1$ is some constant.
We show that the problem admits
an $O(\log d)$-approximation under some mild assumptions on the slack and minimum window size.

Our  approach builds on a formulation of the problem as a configuration LP and involves two main phases:
\emph{a maximization phase} and \emph{residual phase}. Informally, a configuration is a subset of jobs that can be feasibly assigned to a host at a given time slot $t$.
For each $t$, we choose  $O(m\log d)$ configurations with probabilities proportional to their LP-values.  In this phase, jobs may be allocated the resource only for part of their processing length.
In the second phase, we construct a residual instance based on the amount of time each job has been processed. A key challenge is to show that, for any time window $\chi$, 
the total ``area'' of jobs left to be scheduled is at most $1/d$
of the original total area. We use this property to solve the residual instance. We start by describing the configuration linear program that is at the heart of our algorithm.

\subsection{Configuration LP}

Let $J_t\subseteq J$ denote the set of all jobs $j$ such that  $t\in \chi_j$, i.e., $j$
can be allocated resources at time slot
 $t$. For any $t\in [T]$ and $S\subseteq J_t$,  $C= (S, t)$ is  a valid \emph{configuration} on a single host if $\sum_{j\in S} \bar{s}_j \leq \bar{1}^d$, i.e.,
 the jobs in $S$ can be feasibly allocated their resource requirements on a single host at time slot
 $t$. Denote the set of all valid configurations at time $t$ by $\mathcal{C}_t$, and by $\mathcal{C}^j$ the set of all valid configurations
 $(S,t)$,
 such that $S$ contains job $j$. Denote by $x_C$ the indicator variable for choosing configuration $C$, and by $m$ the number of hosts needed to schedule all jobs. The fractional relaxation of the Integer Program formulation of our problem is given below.

\begin{equation*}
\begin{array}{llll}
\text{\bf Primal}: & \text{Minimize } &m  & \nonumber\\
& \text{Subject to:} & m -\sum_{C\in \mathcal{C}_t} x_C \geq 0, &
\forall t\in [T] \nonumber  \\
&& -\sum_{C\in \mathcal{C}^j\cap \mathcal{C}_t} x_C \geq -1, &
\forall j\in J,\,\, t\in [T]  \nonumber\\
&& \sum_{C\in \mathcal{C}^j} x_C \geq p_j, &
\forall j\in J \nonumber\\
&& x_C\geq 0, &
\forall C \nonumber
\end{array}
\end{equation*}

The first constraint ensures that we do not pick more than $m$ configurations for each time slot $t\in [T]$.  The second constraint guarantees that at most one configuration is chosen for each job $j$ at a given time $t$.
Finally, the last constraint guarantees that each job $j$ is
allocated the resource for $p_j$ time slots, i.e., job $j$ is completed.

Given that the LP has an exponentially number of variables, we consider solving the dual.

\[
\begin{array}{llll}
\text{\bf Dual:}&
 \text{Maximize }& \sum_{j\in J}(p_j\alpha_j -\sum_{t\in [T]}\beta_{j,t})& \nonumber\\
& \text{Subject to:} & \sum_{j\in S} (\alpha_j - \beta_{j,t}) - \gamma_t\leq 0, &
\forall C=(S,t) \nonumber  \\
&& \sum_{t\in [T]}\gamma_t \leq 1, & \nonumber \\
&&\alpha_j, \beta_{j,t}, \gamma_t \geq 0&\nonumber\\
&\nonumber
\end{array}
\]

The proof of the following theorem is similar to a result due to Fleischer~\etal~\cite{FGMS11}, with some differences due to the
``negative'' terms in the objective of the dual program.

\begin{theorem}
\label{thm:lp-config}
For any $\epsilon>0$, there is a polynomial time algorithm that yields a $(1+\epsilon)$-approximate solution for the configuration LP.
\end{theorem}
\begin{proof}
We first describe a separation oracle for the dual program. Given a vector $(\gamma_t:t\in [T], \beta_{j,t}: j\in J, t\in [T], \alpha_{j}, j\in J)$, the oracle should either report that the solution
is feasible or find a violating constraint. Clearly, the non-trivial task is to check the exponentially many constraints corresponding to the configurations. To this end, we compute a configuration
$C=(S,t)$, for each $t$, that maximizes the value $\sum_{j\in S} (p_j\alpha_j - \beta_{j,t})$. Subsequently, we can compare this value against the fixed value $\gamma_t$. Further, observe that such a subset $S$ (for a given $t$) can be approximately found by solving the following instance of the \emph{multi-dimensional knapsack} problem. Indeed, we have an item for each job $j \in J_t$ with size $\mathbf{s}_j$ and profit
$p_j\alpha_j - \beta_{j,t}$. The goal is to find a subset of items of maximum total profit
that can fit into the $d$-dimensional bin $1^d$. There is a well known PTAS for the multi-dimensional knapsack problem, for any fixed $d >1$~\cite{FC84}.
Let $\epsilon' > 0$ be the error parameter for the PTAS.

We now run the ellipsoid algorithm on the dual program. We perform a binary search over the  possible values of $z^* \leq  \sum_{j\in J} (\alpha_j -\sum_{t \in T}\beta_{j,t})$.
Suppose that the ellipsoid algorithm reports failure at the value $z^*+\delta$, where $\delta$ is the accuracy parameter of the binary search, which can be made as small as desired. Clearly, the optimal solution value must be lower than $z^*+\delta$. On the other hand, since we are using a $(1+\epsilon')$-approximation oracle, we have that if a solution $(\alpha_j, \beta_{j,t}, \gamma_t)$ is
reported to be feasible, then we must have that $(\alpha_j/(1+\epsilon'), \beta_{j,t}/(1+\epsilon'), \gamma_t)$ is feasible for the original dual program. Hence, the optimal solution lies in the range
 $[z^*/(1+\epsilon'), z^*+\delta]$.

Further, we look at the constraints checked by the ellipsoid algorithm when  the value is set to be $z^*+\delta$. There are polynomial number of such constraints before the algorithm reports a failure. We consider the dual of this restricted LP that is equivalent to a restricted original configuration LP obtained by setting the variables corresponding to the constraints not considered by the ellipsoid algorithm. As noted by Fleischer~\etal~\cite{FGMS11}, the cost of this LP is at most $z^*+\delta$  by LP duality. Thus, for appropriate selection of $ \epsilon' \in (0, \epsilon)$, we obtain a $(1+\epsilon)$-approximate solution for the configuration LP.
\end{proof}

\begin{algorithm}
\caption{Algorithm for the \MM\ problem }
\label{algm:MinM}
\begin{algorithmic}[1]
    \Require{\MM\ instance $(J,\mathcal{W})$, where $J$ is the set of jobs and $\mathcal{W}$ is the set of time-windows}
\Ensure{An assignment of all jobs in $J$ to a (approximate) minimum number of hosts}
\State Solve the configuration LP approximately using Theorem~\ref{thm:lp-config} to obtain a (fractional) solution $\{\hat{x}_C: \forall  C\}$. Let $m$ denote the objective value of the solution (assume $m$ is rounded up to the nearest integer).
\State Let $c>1$ be a constant (to be determined). For each $t\in [T]$ repeat: for $m_1 = cm\log d$ iterations,  choose a configuration $C\in \mathcal{C}_t$ with probability
$\hat{x}_C/m$. If two configurations $C_1 = (S_1,t)$ and $C_2 = (S_2,t)$ where $S_1\cap S_2 \neq \phi$ are chosen in this process, replace $C_2$ with $C_2' = (S_2\setminus S_1, t)$. Continue with this
replacement process until all configuration sets corresponding to a given time-slot $t$ are disjoint.\label{mm:rand}
\State If a configuration $C = (S,t)$ is chosen in Step~\ref{mm:rand} in iteration $i$ for time $t$, assign the jobs in $S$ to host $i$ at time $t$.
\State Let $J_{res}$ be the set of residual jobs constructed from $J$ as follows: for each $j\in J$, let $n_j$ denote the number of configurations containing job $j$ that are chosen in Step~\ref{mm:rand}. Now, define a residual job $j'$ with attributes $(p_{j'} = \max(p_j-n_j,0), s_{j'}= ||\bar{s}_j||_\infty, \chi_{j'} =\chi_j )$. For each job associate a set, {\em forb}$(j')$, of forbidden time-slots in which  a configuration containing job $j$ is already chosen.
\State Use the construction given in Section~\ref{sec:MT_general} to transform $\mathcal{W}$  into a laminar family of intervals $\mathcal{L}$ and obtain the corresponding mapping $\mathfrak{L}:\mathcal{W}\rightarrow \mathcal{L}$. Consider a modified instance $(J'_{res} = \{j': j\in J_{res}\}, \mathcal{L})$ such that $\chi_{j'} = \mathfrak{L}(\chi_j)$.\label{mm:res}
\State Use Lemma~\ref{cor:group-packing} to compute a feasible schedule for the residual jobs $J'_{res}$. Let $m_2$ be the number of additional hosts needed in this step.
\State  Output the resulting assignment of the jobs in $J$ to $m_1+m_2$ hosts.
\end{algorithmic}
\end{algorithm}
\negA

\subsection{The Algorithm}
Let $m^*$ denote  the optimal value of the configuration LP, and let $m \le (1+\epsilon)m^*$ be the
objective value of the approximate solution of the configuration LP, rounded up to the nearest integer.
The detailed description of the algorithm is given in algorithm~\ref{algm:MinM}.
We use the following
two stage process. In the first stage, we choose $O(m\log d)$ configurations  $C$ (with probabilities proportional to their $x_C$ values in the LP solution) for each time slot.
Indeed, this random selection may lead to {\em partial} execution of some jobs $j$, which are allocated the resources for less than $p_j$ time slots.
The second stage amends this, by considering the ``residual'' job parts and assigning them to a set of $m$ new hosts.

Our key technical result (in Lemma~\ref{lem:residual}) is that, with high probability, the total volume of the residual jobs to be scheduled in any time window $\chi$ is sufficiently
 small.  Some additional challenges arise due to the fact that
the time slots used to schedule a job in the first and second stages must be disjoint. Thus, the time slots used for job $j$ in the first stage become ``forbidden'' for $j$  in the second stage.
We associate with each job $j$ a subset of ``forbidden'' time slots, denoted $\text{forb}(j)\subseteq \chi_j$. Any feasible solution for the residual jobs must  ensure that job $j$ is not scheduled at time
 $t\in \text{forb}(j)$. To resolve this issue we need to refine Algorithm~\ref{algm:schedule} (in Lemma~\ref{cor:group-packing}).
 
 \subsection{Analysis}
Towards analyzing our algorithm, we prove several technical lemmas.
The following useful result is due to McDiarmid~\cite{mcdiarmid1989method}.
\begin{lemma}[McDiarmid]
\label{mcdiarmid}
Let $Y = (Y_1, Y_2, \ldots, Y_l)$ be a family of independent random variables, such that $Y_j$ has its domain over the event set $E_j$. Further, suppose $f:(E_1\times E_2\times \ldots \times E_l)\rightarrow \mathbb{R}$ is a real function over the $l$-dimensional event space, such that: $|f({\bar y})- f({\bar y}')| \leq c_i$
for all ${\bar y}$ and ${\bar y}'$ differing in exactly the $i$th coordinate. Then,
$Pr[f(Y) -\mathbb{E}(f(Y)) \geq \psi] \leq \texttt{exp}({-2\psi^2/\sum_{i=1}^lc_i^2}) $.
\end{lemma}

The next lemma gives the conditions implying that the total volume of the residual jobs is {\em small}. This is essential for obtaining our performance bounds.
\begin{lemma}
\label{lem:residual}
Let  $\chi$ be any sub-interval of $[T]$.
For any $\epsilon\in (0,1)$, $\omega\in (0,1)$, and a sufficiently large value of $\theta$ that depends on $\epsilon$ and $\omega$,
the following holds with probability at least $1-\epsilon$
\begin{equation}
\label{eq:residual_chi_volume_bound}
\sum_{j\in J_{res}: \chi_{j}\subseteq \chi} p_j' \lVert \bar{s}_j\rVert_\infty \leq \omega m|\chi|,
\end{equation}
assuming interval $\chi$ staisfies
\begin{equation}
\label{eq:large_windows_assumption}
|\chi| \geq \frac{1}{m}\theta  d^2\log d \log (\epsilon^{-\frac{1}{2}}T).
\end{equation}
Further, if $T = O(d^\delta)$, for some constant $\delta \geq 0$, then the restriction on the length of $\chi$ in
(\ref{eq:large_windows_assumption}) can be dropped and (\ref{eq:residual_chi_volume_bound}) holds
for any interval $\chi$.
\end{lemma}

\begin{proof}

    Consider an interval $\chi$ that satisfies (\ref{eq:large_windows_assumption}) and a job $j$ such that $\chi_j \subseteq \chi$.
    Define $X_{jt}$ as the total (fractional) value of configurations corresponding to job $j$ and time $t$, i.e., $X_{jt} = \sum_{C\in \mathcal{C}^j\cap \mathcal{C}_t}\hat{x}_C$. Then,
\begin{equation}
\label{eqn:pj} \sum_{t\in \chi_j} X_{jt} = \sum_{C\in \mathcal{C}^j} \hat{x}_C  \geq p_j
\end{equation}
Further, by the LP constraints, we have that $X_{jt}\leq 1$.

For the  analysis, we partition $\chi$ into $p_j$ regions $R_1^j, R_2^j, \ldots, R_{p_j}^j$, such that
\begin{equation}
\label{eq:def_regions}
\sum_{t\in R_k^j}X_{jt} \geq 1/2.
\end{equation}
The regions are not necessarily formed of consecutive time slots in $\chi$.
One way of doing this is as follows. Start with singleton regions $\{t\}$ for which $X_{jt} \geq 1/2$. Let $q_j$ denote the number of regions generated.
If $q_j\geq p_j$ then we are done; otherwise, until  $p_j$ regions are obtained, we do the following, for $k\in [q_j+1,p_j]$: starting with $R_k^j = \phi$ and while  $\sum_{t\in R_k^j}X_{jt} < 1/2$, add a time slot $t'$ that has not yet been assigned to any region. Since any such $t'$ has not been chosen as a singleton region, we have $X_{jt'} < 1/2$;
 therefore,  $\sum_{t\in R_k^j}X_{jt} \leq 1$. From Equation~(\ref{eqn:pj}) we have that at least $p_j$ regions are generated in this process.

Let $\mathbf{R}_k^j$ denote the event that job $j$ is \emph{not} scheduled in a given region $R_k^j$ in the first stage of the algorithm. We now compute the probability $\mathbb{P}(\mathbf{R}_k^j)$. Let $Pr(j,t)$ be the probability that $j$ is not assigned to any host at time $t$. Then,
$$Pr(j,t) = \left(1-\frac{X_{jt}}{m}\right)^{cm\log d} \leq \texttt{exp}({-cX_{jt}\log d})$$
Hence, the probability that job $j$ is \emph{not} scheduled in  region $R_k^j$ satisfies

\begin{align}
\mathbb{P}(\mathbf{R}_k^j)& \leq \displaystyle{\Pi_{t\in R_{k}^j}}\texttt{exp}({-cX_{jt}\log d}) = \texttt{exp}({-c\log d\sum_{t\in R_k^j} X_{jt}}) \nonumber \\
& \leq \texttt{exp}({-\frac{c}{2}\log d}) = d^{-\frac{c}{2}}
\label{eq:prob_j_notin_Rkj}
\end{align}
\remove{
\begin{equation}
\label{eq:prob_j_notin_Rkj}
\mathbb{P}(\mathbf{R}_k^j)\leq \Pi_{t\in R_{k}^j}\texttt{exp}({-cX_{jt}\log d}) = \texttt{exp}({-c\log d\sum_{t\in R_k^j} X_{jt}}) \leq \texttt{exp}({-\frac{c}{2}\log d}) = d^{-\frac{c}{2}}
 \end{equation}
}
The last inequality follows from (\ref{eq:def_regions}).

Recall that $p_j'$ (in Algorithm~\ref{algm:MinM}) is the processing time required for the residual job $j$. The  value of $p_{j}'$ is upper bounded by the number of regions in which job $j$
 has {\em not} been scheduled. Therefore, we have
$\mathbb{E}[p_j'] \leq \sum_{k\in [p_j]} \mathbb{P}(\mathbf{R}_k^j) \leq p_jd^{-\frac{c}{2}}$.
By the LP constraints, we have that $\sum_{j:\chi_j \subseteq \chi} p_j\bar{s}_j \leq m|\chi|\bar{1}^d $.
A simple (folklore) fact
is that if a set $V$ of  $d$-dimensional items can be vector-packed into the cube $\bar{1}^d$, then the set of 1-dimensional items obtained from $V$ by taking the $\ell_\infty$
norm of item size vectors can be bin packed into at most $d$ unit sized bins.
 It follows that $\sum_{j:\chi_j \subseteq \chi} p_j\lVert \bar{s}_j\rVert_\infty \leq m|\chi| d$. Hence, using (\ref{eq:prob_j_notin_Rkj}), we obtain the following upper bound on the expected
 total volume of residual jobs in $\chi$:
 \begin{equation}
\label{eqn:expectedvalue}
\mathbb{E}\left[\sum_{j: \chi_j\subseteq \chi} p_j'\lVert\bar{s}_j\rVert_\infty\right] \leq m|\chi|d^{1-\frac{c}{2}}
\end{equation}

To complete the proof, we need to show that, with high probability, the total volume of residual jobs in {\em any} sub-interval of $[T]$ is small. To this end, we apply Lemma~\ref{mcdiarmid} as follows. For each $t\in \chi$ and iteration $i\in [cm\log d]$, there is an associated random variable
indicating the configuration $C$ chosen in iteration $i$ at time $t$.
Note that since Lemma~\ref{mcdiarmid} requires independence between events, the configurations considered are the ones prior to the modifications applied in Step~\ref{mm:rand}.
This is valid as we define the function $f(\cdot)$ in Lemma~\ref{mcdiarmid} to take into account the modifications applied in Step~\ref{mm:rand}.
This function is defined  as the quantity $A_{res}[\chi] = \sum_{j: \chi_j\subseteq \chi} p_j'\lVert\bar{s}_j\rVert_\infty$ which is a function of these $cm|\chi|\log d $ independent random variables.
 If one of these variables is altered, it might affect the choice of at most one configuration, namely,
 a configuration $C'$ selected in iteration $i$ at time $t$ is replaced by another configuration $C''$.
 Suppose $A'_{res}(\chi)$ and $A''_{res}(\chi)$ are the two corresponding realizations of the random variable $A_{res}(\chi)$. We now bound the quantity $|A'_{res}(\chi) - A''_{res}(\chi)|$: the worst
 scenario is clearly when none of the jobs in configuration $C'$ is contained in any other configuration chosen at time $t$, whereas every job in $C''$ is contained in some other configuration chosen
 at time $t$; or vice versa. Thus, we have  $|A'_{res}(\chi) - A''_{res}(\chi)| \leq \max(\sum_{j\in C'}\lVert\bar{s}_j\rVert_\infty, \sum_{j\in  C''}\lVert\bar{s}_j\rVert_\infty) \leq d$.
Applying Lemma~\ref{mcdiarmid}, we have, for any $\omega' \geq 0$

\begin{align}
Pr[A_{res}(\chi) -\mathbb{E}(A_{res}(\chi)) \geq \omega' m|\chi|] & \leq
\frac{\texttt{exp}({-(2{\omega'}^2)m^2|\chi|^2}}{cm \log d |\chi| d^2)} \nonumber \\
& = \texttt{exp}({-\frac{2{\omega'}^2}{c}\cdot \frac{m|\chi|}{d^2\log d}}) \nonumber
\end{align}



Let $c > 2$ be a sufficiently large constant such that $\omega > d^{1-\frac{c}{2}}$. We set $\omega' = \omega - d^{1-\frac{c}{2}}$, and
$\theta = c{\omega'}^{-2}$. Then, we have
\begin{align*}
Pr[A_{res}(\chi) & \!\geq \omega m|\chi|] = \displaystyle{Pr[A_{res}(\chi) \geq (\omega' + d^{1-\frac{c}{2}}) m|\chi|] }\\
& \leq \displaystyle{Pr[A_{res}(\chi) -\mathbb{E}(A_{res}(\chi)) \geq \omega' m|\chi|]}\\
& \leq \!\displaystyle{\texttt{exp}({-\frac{2 \omega'^2}{c}\cdot \frac{m|\chi|}{d^2 \log d}})}
= \displaystyle{ \texttt{exp}({\frac{-2}{\theta}\cdot \frac{m|\chi|}{d^2 \log d}})}\\
& \!\leq \displaystyle{ \texttt{exp}({-2 \log (\epsilon^{-\frac{1}{2}}T)}) = \epsilon T^{-2}}
\end{align*}
The first inequality follows from (\ref{eqn:expectedvalue}) and the third from (\ref{eq:large_windows_assumption}).
 Since the total number of distinct intervals possible in $[T]$ is at most $T^2$, by applying the union bound, the probability that some interval $\chi$ that satisfies (\ref{eq:large_windows_assumption})  fails to satisfy (\ref{eq:residual_chi_volume_bound}) is at most $\epsilon$.

Now, consider  the case where $T = d^\delta$, for some constant $\delta \geq 0$ (independent of $d$). We have
$$Pr[A_{res}(\chi) \geq \omega m|\chi|] \leq Pr[A_{res}(\chi) \geq \frac{\mathbb{E}(A_{res}(\chi)) \omega }{d^{1-\frac{c}{2} }}]  \leq \frac{d^{1-\frac{c}{2}}}{\omega}.$$
The first inequality follows from (\ref{eqn:expectedvalue}) and the second from Markov's inequality. As before, we observe that total number of distinct intervals possible in $[T]$ is at most $T^2 \leq d^{2\delta}$. Now, we choose
$\epsilon'>0$ to be a large enough constant satisfying $\omega^{-1}d^{-\frac{\epsilon'}{2}}\leq \epsilon$, and $c \geq 2(2 \delta +1) + \epsilon'$.
Hence,
$$
Pr[A_{res}(\chi) \geq \omega m|\chi|]  \leq \frac{d^{1 -(2 \delta +1) - \epsilon'/2}}{\omega} \leq \epsilon T^{-2}.
$$
Again, since the total number of distinct intervals possible in $[T]$ is at most $T^2$, by applying the union bound, the probability that some interval $\chi$ fails to satisfy
(\ref{eq:residual_chi_volume_bound}) is at most $\epsilon$.
\end{proof}

The next lemma gives the condition that guarantee feasible schedule for the residual jobs.

\begin{lemma}
\label{cor:group-packing}
Let $(J,\mathcal{L})$ be
a single resource
laminar instance  of {\em \MM} such that for any job $j$, $p_j + |\text{forb}(j)| \leq \lambda |\chi_j|$,
for some $\lambda \in (0,1)$, and for any
$\chi \in \mathcal{L}$: $\sum_{\{j\in J: \chi_j \subseteq \chi \} } p_j s_j \leq \alpha m|\chi|$,
for some $\alpha \in (0,1)$. Then,
if $\lambda + 2\alpha \leq 1$, all jobs in $J$ can be assigned to $m$ hosts.
\end{lemma}

\begin{proof}
We use an algorithm similar to Algorithm~\ref{algm:schedule}. The only difference is that in our case
$avail(j) \subseteq \chi_j\setminus \text{forb}(j)$.
Similar to the proof of Theorem~\ref{thm:group-packing} we prove that the algorithm never reports {\bf fail}
while assigning job $j$. Suppose that $j$ was assigned to $p_j' < p_j$ bins before this {\bf fail}.
For $t \in \chi_j$, we say that bin $(i,t)$
is \emph{bad} if either $t \in \text{forb}(j)$, or $(i,t)$ is colored gray, or $j$  has been assigned to some bin $(i',t)$ in the same time slot. It can be shown by induction that as long as no job
$j^+$ such that $\chi_j\subset \chi_{j^+}$ has been allocated the resource: the number of bad bins while processing job $j$ is at most $ \lambda m |\chi|$.
Since we pair the black bins $(i,t)\leftrightarrow (i',t')$ only if $\sum_{j\in S_{(i,t)}}s_j$ + $ \sum_{j'\in S_{(i',t')}}s_{j'} > 1$, the total number of black bins $< 2\alpha m|\chi_j|$.
Hence, the total number of bins that are black or bad is $< (\lambda + 2\alpha) m|\chi_j|$.
Thus if $\lambda+2\alpha \le 1$, there should be at least one bin $(i^*,t^*)$ that is neither black nor bad.
But in this case, we could have assigned $j$ to host $i^*$ at time $t^*$, which is a contradiction to the assumption that the algorithm reports a {\bf fail}.
\end{proof}

The above lemmas lead to an $O(\log d)$ performance guarantee for instances with {\em large} time windows, as formalized in the next result.
\begin{theorem}
\label{thm:crcsd}
Let $(J,\mathcal{W})$ be an instance of {\em \MM } with slackness parameter $\lambda \in (0,\frac{1}{4})$.
Fix an $\epsilon \in (0,1)$. If
$|\chi_j|\geq \frac{1}{m}\theta d^2 \log d \log (T\epsilon^{-\frac{1}{2}})$  $\forall~j\in J$,
for sufficiently large constant $\theta$, then
Algorithm~\ref{algm:MinM} yields an $O(\log d)$ approximation guarantee with probability at least $1-\epsilon$.
\end{theorem}
\begin{proof}
The optimal objective value of the configuration LP, denoted $m^*$ is a lower bound on the number of hosts required for the instance  $(J,\mathcal{W})$.
By Theorem~\ref{thm:lp-config}, $m\leq (1+\epsilon)m^*$. Now, assuming Algorithm~\ref{algm:MinM} is correct,
the number of hosts used is at
 most $m(c\log d) + m$, implying an $O(\log d)$ approximation guarantee.
 Below, we prove correctness of Algorithm~\ref{algm:MinM}, i.e., we show that it feasibly schedules all the jobs in $J$.

It suffices to show
that the residual set of jobs $J_{res}$ can be successfully scheduled.
Recall that Algorithm~\ref{algm:MinM} uses the construction given in Section~\ref{sec:MT_general} to transform $\mathcal{W}$ into a laminar set of intervals $\mathcal{L}$ and to obtain a mapping $\mathfrak{L}:\mathcal{W}\rightarrow \mathcal{L}$.
Then, the algorithm computes a schedule of the residual set of jobs $J_{res}$
by solving the laminar instance $(J'_{res}= \{j': j\in J\},  \mathcal{L})$, where $\chi_{j'} = \mathfrak{L}(\chi_j)$.
Observe that
for any job $j\in J_{res}$ we have that $p_j' + |\text{forb}(j)| = p_j \leq \lambda|\chi_j| \leq 4\lambda |\chi_{j'}|$,
where the last inequality follows from Lemma~\ref{lem:laminarlpbound}.
By the same lemma, we also get that for any job $j\in J_{res}$,
$|\chi_{j'}|\geq \frac 14 |\chi_j|$. Now, let $\tilde{\chi}_{j'}$ be the union of all the time windows mapped by $\mathfrak{L}$ to time windows in $\chi_{j'}$. Also, by Lemma~\ref{lem:laminarlpbound}
$|\tilde{\chi}_{j'}| \leq 4|\chi_{j'}|$. Clearly, $|\tilde{\chi}_{j'}| \geq
|\chi_{j'}| \geq \frac{1}{4m}\theta d^2 \log d \log (T\epsilon^{-\frac{1}{2}})$.
Now, apply Lemma~\ref{lem:residual} to the intervals $\tilde{\chi}_{j'}$ for $j\in J_{res}$ (with appropriate value of $\theta$), to get that
that $\sum_{j\in J_{res}:\chi_{j'}\subseteq \chi} p_j' \lVert\bar{s}_j\rVert_\infty \leq 4\omega m|\chi|$ with probabillty at least $1- \epsilon$.
We now apply Lemma~\ref{cor:group-packing}, by setting $\alpha \leftarrow 4\omega$ and $\hat{\lambda} \leftarrow 4\lambda$.
 We note that there is a feasible schedule of jobs $j'$ on $m$ hosts if $4\lambda + 8\omega = 1$. Indeed, for any $\lambda < \frac{1}{4}$, there is a positive constant $\omega$ that satisfies this equation. Finally, it is easy to show that transforming the instance back to $d$ dimensions (by replacing the requirement of $\lVert\bar{s}_j\rVert_\infty$ by $\bar{s}_j$) the schedule remains feasible. This completes the proof.
\end{proof}

We conclude our analysis with the following result.
\begin{theorem}
\label{cor:noassumption}
Let $(J,\mathcal{W})$ be an instance of {\em \MM} with slackness parameter $\lambda \in (0,\frac{1}{4})$.
Fix an $\epsilon \in (0,1)$.
There is a polynomial time algorithm that yields an
$O(\log d\log^*T)$ approximation ratio with probability at least $1-\epsilon$.
\end{theorem}
\begin{proof}
The key idea is the following.
Starting with the maximum schedule length $T$, we recursively define $\kappa +1$ ranges for the time-window sizes in the original instance. We then partition the set of
jobs $J$ to $\kappa +1$ subsets, each containing jobs with time windows within the corresponding range, where $\kappa = O(\log^*T)$.
The crux of this partition is that  the resulting $\kappa +1$ instances of our problem satisfy
the conditions of Theorem~\ref{thm:crcsd}. In particular, all jobs have `large' windows. Thus, we can obtain for each instance a $\log d$-approximate solution.
Formally, let $\theta$ be the constant in Theorem~\ref{thm:crcsd}. Set $\gamma = \theta d^2\log d$ (we assume that the $\log$ is base 2).
Define the function $\psi:\mathbf{N}\rightarrow \mathbf{R^+}$ as follows:
$$
\psi(i)=
\begin{cases}
    4\lceil{\gamma^2}\rceil  & \mbox{ if } i =1 \\
    \min \{T,\left(2^{1/(2\gamma)}\right)^{\psi(i-1)}\} & \mbox{ if } i>1
\end{cases}
$$
It is easy to verify that, for $i\geq 1$, we have
$\psi(i) \le \psi(i+1)$, and $\psi(i) \ge 2\gamma \log_2 \psi(i+1)$. Also,
let $\kappa$ be the smallest integer for which $\psi(\kappa) = T$. Then, we have $\kappa = O(\log^* T)$.

We partition the set of intervals $\mathcal{W}$ into groups based on their length as follows:
$$\mathcal{I}_0 = \{\chi: |\chi|\in [1, \psi(1)]\},$$
$$\mathcal{I}_w = \{\chi: |\chi|\in (\psi(w), \psi(w+1)]\}\quad \mbox{ for all }w\in [1, \kappa].$$

Next, we define $\kappa +1$ instances of our problem, where the $w^{th}$ instance is given by:
$$(J_w = \{j\in J: \chi_j \in \mathcal{I}_w\}, \mathcal{I}_w).$$

Since each of the above instances requires to schedule a subset of jobs in the original instance, they optimally need at most $m$ hosts to complete all jobs. Consider the $w^{th}$ instance. The
largest window size here is at most $\psi(w+1)$. We further partition this instance as follows. Let
$$ \mathcal{I}_{w,Odd}^i =[(2i+1)\psi(w+1)+1, \min\{(2i+3)\psi(w+1),T\}],$$
where $i\geq 0 \text{ and } (2i+1)\psi(w+1) < T$. Similarly, let
$$ \mathcal{I}_{w,Even}^i =[(2i)\psi(w+1)+1, \min\{(2i+2)\psi(w+1),T\}],$$
where $i\geq 0 \text{ and } (2i)\psi(w+1) < T$.
Now, we define
$$J_{w,Odd}^i = \{j \in J_w: \chi_j \subseteq \mathcal{I}_{w,Odd}^i \}$$
and
$$J_{w,Even}^i = \{j \in J_w: \chi_j \subseteq \mathcal{I}_{w,Even}^i \}.$$
Let $J_{w,Odd}=\cup_i J_{w,Odd}^i $ and $J_{w,Even}=\cup_i J_{w,Even}^i $.
Finally, remove each job $j\in J_{w,Odd}\cap J_{w,Even}$ from $J_{w,Even}$ and the corresponding $J_{w,Even}^i$.
Consequently, $J_{w,Odd}\cap J_{w,Even} =\emptyset$.

For any $w>0$, fix an $i\ge 0$ such that $(2i+1)\psi(w+1) < T$ and
consider the instance defined by the jobs in  $J_{w,Odd}^i$. We claim that this instance
can be solved using Theorem~\ref{thm:crcsd}.
Indeed, the total number of time slots in this instance is $2\psi(w+1)$, and the time-window of any job $j$ in the instance is $|\chi_j| \geq \psi(w) \geq 2\theta d^2\log d \log \psi(w+1)$. Thus, the conditions in Theorem~\ref{thm:crcsd} are satisfied, and we can obtain a feasible schedule using $O(m\log d)$ hosts.
Now, fix an $i\ge 0$ such that $(2i+1)\psi(1) < T$ and  consider
the instance defined by the jobs  in  $J_{0,Odd}^i$. We claim that this instance
can also be solved using Theorem~\ref{thm:crcsd}, since in this case the total number of time slots in this instance is
$2\psi(1)=O(d^\delta)$, for some constant $\delta \geq 0$.
\remove{
Note that for any $w$, the odd instances  $J_{w,Odd}^i, J_{w,Odd}^\ell$  for $i \neq \ell$ are mutually disjoint (all jobs and time-windows are completely disjoint). Thus,
we can solve them in parallel using $O(m\log d)$ (say $m_o$) hosts and merge them together. We can argue the same about $J_{w,Even}$. Suppose we need $m_o$ and $m_e$ hosts to solve $J_{w,Odd}$ and $J_{w,Even}$, respectively.  Since no job is shared between $J_{w,Odd}$ and $J_{w,Even}$, we can schedule all the jobs in $J_w$ using $m_o + m_e$ hosts.
Thus, for any $w$, we obtain a feasible schedule using $O(m\log d)$ hosts.
}

Note that for any $w$, the odd instances  $J_{w,Odd}^i, J_{w,Odd}^\ell$  for $i \neq \ell$ are mutually disjoint (jobs and time-windows). Thus,
we can solve them in parallel using the same $O(m\log d)$ hosts. We can do the same for $J_{w,Even}$. Suppose we need $m_o$ and $m_e$ hosts to solve $J_{w,Odd}$ and $J_{w,Even}$, respectively.  Since no job is shared between $J_{w,Odd}$ and $J_{w,Even}$, we can schedule all the jobs in $J_w$ using
$m_o + m_e=O(m\log d)$ hosts.

Now, to handle the instances corresponding to all $w\in \{0, \ldots , \kappa\}$, we note that no job is shared among the instances. Therefore, we can aggregate the hosts to obtain a feasible schedule for all instances using $O(m\kappa \log d)$ hosts.
\end{proof}

\bibliography{rap}

\end{document}